\documentclass[12pt,draftcls,onecolumn,english]{IEEEtran}

\IEEEoverridecommandlockouts

\usepackage{amsmath,amsfonts, amssymb,color}

\newtheorem{theorem}{Theorem}
\newtheorem{corollary}{Corollary}
\newtheorem{lemma}{Lemma}
\newtheorem{remark}{Remark}
\newtheorem{definition}{Definition}
\newtheorem{example}{Example}

\newcommand{\Poly}{{\textbf{P}}}
\newcommand{\NP}{{\textbf{NP}}}
\newcommand{\coNP}{{\textbf{coNP}}}
\newtheorem{proof}{Proof}

\usepackage{epsfig} 
\usepackage{psfrag}
\usepackage{subfigure}
\usepackage{enumerate}
\usepackage{times} 
\usepackage{latexsym}
\usepackage{xspace}
\usepackage{amssymb,amsmath,amsfonts,color}
\usepackage{mathrsfs}
\usepackage{cite}
\usepackage{slashbox}

\title{Robustness of Complex Networks \\ with Implications for Consensus and Contagion 
\thanks{This material is based upon work supported in part by the Natural Sciences and Engineering Research Council of Canada, and by the Waterloo Institute for Complexity and Innovation. Parts of this paper were presented at the 51st IEEE Conference on Decision and Control \cite{zhang2012robustness}.}}

\author{Haotian~Zhang, Elaheh~Fata~and~Shreyas~Sundaram%
\thanks{The authors are with the Department of Electrical and Computer Engineering at the University of Waterloo. E-mail for corresponding author: {\tt ssundara@uwaterloo.ca}.}}

\begin{document}
\maketitle

\begin{abstract}
We study a graph-theoretic property known as robustness, which plays a key role in certain classes of dynamics on networks (such as resilient consensus, contagion and bootstrap percolation). This property is stronger than other graph properties such as connectivity and minimum degree in that one can construct graphs with high connectivity and minimum degree but low robustness.
However, we show that the notions of connectivity and robustness coincide on common random graph models for complex networks (Erd\H os-R\'enyi, geometric random, and preferential attachment graphs).  More specifically, the properties share the same threshold function in the Erd\H os-R\'enyi model, and have the same values in one-dimensional geometric graphs and preferential attachment networks. This indicates that a variety of purely local diffusion dynamics will be effective at spreading information in such networks.  Although graphs generated according to the above constructions are inherently robust, we also show that it is $\coNP$-complete to determine whether any given graph is robust to a specified extent.
\end{abstract}


\maketitle
\thispagestyle{empty}
\pagestyle{empty}

\section{Introduction}
\label{sec:intro}
The emergence of collective dynamics in large networks of interacting agents has inspired the study of complex networks; classical examples of such networks abound in both the natural world (e.g., ecological systems, biological systems, and social systems), and in engineered applications (e.g., the Internet, the power grid, large-scale sensor networks).   Due to their prevalence, a topic of particular interest has been the {\it robustness} of such networks to disruptions, both in the structure and in the dynamics that are occurring on the network.  Studies of structural robustness have characterized the ability of different networks to withstand the loss of nodes, either due to accidental failure or to targeted attack \cite{Albert00,Callaway00,Bollobas04}, with the connectivity of the network being of primary interest.
Studies of the robustness of dynamics, on the other hand, focus on the ability of the nodes in the network to achieve certain objectives even when some nodes deviate from expected behavior.  For instance, various dynamics for information diffusion have been studied in the context of synchronization (or consensus)  \cite{Lynch96}, information cascades \cite{Morris00,Goldenberg01}, and broadcasting \cite{Hromkovic05}.  In such cases, a fundamental challenge is to identify topological properties that allow legitimate information to propagate throughout the network, while limiting the effects of illegitimate actions.
In the case of distributed consensus in computer networks, algorithms have been proposed to overcome adversaries when all nodes know the network topology and the {\it connectivity} of the network is sufficiently high \cite{Lynch96,SundaramHadjicostis11,Pasqualetti12}. 
However, it was shown in \cite{ZhangSundaram2012,leblanc2013resilient} that connectivity is no longer sufficient to guarantee resilient consensus when the nodes use a natural class of algorithms that only require each node to know its own neighborhood.
Instead, \cite{ZhangSundaram2012} introduced a definition of network robustness to deal with such dynamics and showed that resilient consensus can be reached without requiring global information in graphs that are sufficiently robust.
This notion of robustness is stronger than other topological properties such as connectivity (in that highly robust graphs require high connectivity) and as we describe later, also plays a role in the study of contagion and bootstrap percolation in networks.

Motivated by the role that this notion of robustness plays in the above dynamics, in this paper we study this property in three common random graph models for complex networks: Erd\H os-R\'enyi graphs, geometric random graphs, and preferential attachment graphs.  In Erd\H os-R\'enyi graphs, it was established in \cite{ErdosRenyi1961} that the properties of connectivity and minimum degree share the same threshold function; we show that robustness also shares this threshold. 
This is perhaps surprising, given the existence of pathological graphs where robustness and connectivity are far apart (as shown later in this paper), and indicates that the graphs gain a richer structure at this threshold than simply being $r$-connected.
For the other two models,  we show that robustness and connectivity are equivalent in one-dimensional geometric graphs and in certain preferential attachment models. Our results reveal that a variety of diffusion dynamics (that are agnostic of the network structure) will be effective at spreading information in such networks.
Finally, as a counterpoint to the random graph analysis above, we provide a negative result for determining whether any {\it given} graph is $r$-robust, showing that this problem is $\coNP$-complete for any $r \ge 2$.

The rest of this paper is organized as follows. In Section~\ref{sec:motivation}, we provide the definition of robustness that we consider in this paper and give motivating applications. In Sections~\ref{sec:E-R}, \ref{sec:geometric}, and \ref{sec:preatt}, we study the robustness of Erd\H os-R\'enyi graphs, one-dimensional geometric graphs, and preferential attachment graphs, respectively.   In Section~\ref{sec:complexity}, we provide the complexity analysis of the robustness problem in general graphs, and conclude in Section~\ref{sec:summary}.


\section{Robustness of Networks}
\label{sec:motivation}

Consider a network modeled by the {\it undirected} graph $\mathcal{G}=\{\mathcal{V},\mathcal{E}\}$, where $\mathcal{V}=\{1,...,n\}$ is the set of nodes and $\mathcal{E} \subseteq \mathcal{V} \times \mathcal{V}$ is the set of edges in the network. An edge $(i,j)\in \mathcal{E}$ indicates that nodes $i$ and $j$ can communicate with each other.
The set of {\it neighbors} of node $i$ is defined as $\mathcal{V}_i=\{j\in\mathcal{V}\colon (i,j)\in \mathcal{E}\}$, the {\it degree} of node $i$ is denoted by $d_i =| \mathcal{V}_i |$, and the {\it minimum degree} of the network is $\min_{i\in\mathcal{V}}d_i$. 
The {\it connectivity} of the network is the largest integer $r$ such that every pair of nodes has at least $r$ pairwise node-disjoint paths between them; this is a fundamental metric in networks and captures information redundancy across the network through independent paths. By Menger's theorem~\cite{West01}, the connectivity of a network is also equal to the smallest number of nodes that have to be removed in order to disconnect the graph. A graph is $r$-connected if its connectivity is at least $r$. In order to capture another form of information redundancy,  the following topological properties were proposed in \cite{ZhangSundaram2012}.

\begin{definition}[$r$-Reachable Set]
For a graph $\mathcal{G}$ and a subset $\mathcal{S}$ of nodes of $\mathcal{G}$,  $\mathcal{S}$ is an \textbf{$r$-reachable set} if $\exists i\in \mathcal{S}$ such that $| \mathcal{V}_i\setminus\mathcal{S}| \ge r$, where $r\in\mathbb{Z}_{\geq0}$.\footnote{Note that $\mathbb{Z}_{\geq{a}}$ represents the set of integers bigger than or equal to $a$.}
\label{def:r_reachable}
\end{definition}

\begin{definition}[$r$-Robust Graph]
A graph $\mathcal{G}$ is \textbf{$r$-robust} if for every pair of nonempty, disjoint subsets of $\mathcal{V}$, at least one of the subsets is $r$-reachable, where $r\in\mathbb{Z}_{\geq0}$.
\label{def:r_robust}
\end{definition}

In words, a set $\mathcal{S}$ is $r$-reachable if it contains a node that has at least $r$ neighbors outside that set. 
Intuitively, the $r$-reachability property captures the idea that some node inside the set is influenced by a sufficiently large number of nodes from outside. 
While $r$-connectedness implies that given any two disjoint sets, the nodes in at least one of the sets collectively have $r$ neighbors outside, $r$-robustness indicates that there is at least one node in one of the sets that {\it by itself} has $r$ neighbors outside.  Since all graphs are trivially $0$-robust, we will primarily focus on the cases where $r\ge{1}$ in the rest of the paper.

We will be using the following important properties from \cite{leblanc2013resilient} to relate robustness with the concepts of connectivity and minimum degree (the result in \cite{leblanc2013resilient} applies to directed graphs and encompasses undirected graphs as a special case).

\begin{lemma}[\cite{leblanc2013resilient}]
For any $r\in \mathbb{Z}_{\geq0}$, if $\mathcal{G}$ is $r$-robust, then $\mathcal{G}$ is at least $r$-connected and has minimum degree at least $r$.
\label{lem:robust_connectivity}
\end{lemma}

\begin{lemma}[\cite{leblanc2013resilient}]
A graph $\mathcal{G}$ is 1-robust if and only if it is 1-connected.
\label{lem:1_robust_connectivity}
\end{lemma}

The above results show that $r$-robustness is a stronger property than $r$-connectivity (except for the case where $r=1$ and the trivial case where $r=0$). In fact, there exist graphs that are very highly connected but have very low robustness. For example, consider the network shown in Figure~\ref{fig:Counterexample}. Sets $\mathcal{S}_1$ and $\mathcal{S}_2$ have $\frac{n}{2}$ nodes (suppose $n$ is even), and induce complete subgraphs (i.e., each node in each set is connected to all other nodes in its set). Each node has exactly one neighbor from the other set.  
This graph is $\frac{n}{2}$-connected and has minimum degree $\frac{n}{2}$ but is only 1-robust since both $\mathcal{S}_1$ and $\mathcal{S}_2$ are only 1-reachable.

In the rest of this section, we will motivate the study of reachable sets and robustness with several specific examples of diffusion dynamics (resilient consensus, contagion and bootstrap percolation).

\begin{figure}[Counterexample]
\centering
\includegraphics[width=5.5cm]{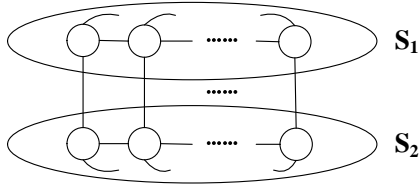}
\caption{Example of a graph that has minimum degree $\frac{n}{2}$ and connectivity $\frac{n}{2}$, but is only 1-robust.  Sets $\mathcal{S}_1$ and $\mathcal{S}_2$ induce complete graphs on $\frac{n}{2}$ nodes, and the edges between the sets form a perfect matching.} 
\label{fig:Counterexample}
\end{figure}

\subsection{Resilient Consensus Using Only Local Information}
\label{subsec:conn}

Consider a setting where each node in the network holds some private information (an opinion, a measurement, etc.). The network operates synchronously, and at each time-step, each normally operating node updates its value (information) as a weighted average of its neighbors' values and its own value. However, there may exist misbehaving nodes which do not follow this pre-specified rule. Under certain fault models (which define the distribution and behavior of the misbehaving nodes), an algorithm is said to achieve {\it resilient asymptotic consensus} if the values held by all normal nodes asymptotically converge to the same value, for any choice of initial values.

As mentioned in the Introduction, the {\it connectivity} of the network has traditionally been viewed as the key metric with regard to resilience of consensus algorithms (and information diffusion algorithms in general).
If the connectivity of the network is $2F$ or less (where $F\in\mathbb{Z}_{\geq1}$), then there exists at least one set of $F$ coordinated misbehaving nodes that can prevent the network from reaching consensus on certain functions of the initial values {\it regardless} of the mechanism that is used to achieve consensus  \cite{Lynch96,Hromkovic05}.  On the other hand, if the connectivity is $2F+1$ or higher, various algorithms have been proposed to overcome misbehaving  nodes under the local broadcast model of communication (i.e., where each misbehaving node is restricted to send the same value to all of its neighbors at each time-step) \cite{Lynch96,SundaramHadjicostis11,Pasqualetti12}.

While the above connectivity bounds provide fundamental limitations on the resilience of networks to misbehaving nodes, the mechanisms proposed to overcome misbehavior typically make the assumption that all nodes know the entire network topology (which is unrealistic in large networks). To remedy this, under the \textbf{$F$-local model} (where there are at most $F$ misbehaving nodes in each normal node's neighborhood and there is no restriction on their behavior), consider the following {\it Weighted-Mean-Subsequence-Reduced (W-MSR)} algorithm: at each time-step, each normal node disregards the largest and smallest $F$ nodes in its neighborhood (breaking ties arbitrarily) and updates its state to be a weighted average of the remaining values.\footnote{We refer to \cite{ZhangSundaram2012,leblanc2013resilient,Azad94,Nitin12,Dolev86} for a more complete description of this and other similar algorithms.} 
To see why connectivity is no longer an appropriate metric for studying such an algorithm, consider again the network in Figure~\ref{fig:Counterexample} and suppose that nodes in $\mathcal{S}_1$ and $\mathcal{S}_2$ have initial values $a$ and $b$, respectively. When $a\ne b$, by using the W-MSR algorithm, each node will throw away the value of its neighbor from the opposite set and thus its own value will remain unchanged, even when there are no misbehaving nodes.  Thus, consensus will not be reached in this network, indicating that even networks with a large degree or connectivity are not sufficient to guarantee consensus under such algorithms.

Taking a closer look at Figure~\ref{fig:Counterexample}, we see that the reason for the failure of consensus in this graph is that no node in either of the two sets receives enough information from {\it outside} its own set. However, if a graph is $r$-robust (for sufficiently large $r$), new information will penetrate at least one out of any two subsets of nodes, preventing stalemates of this form.
The following result from \cite{leblanc2013resilient} formalizes this property.

\begin{theorem}[\cite{leblanc2013resilient}]
Under the $F$-local model, the W-MSR algorithm achieves resilient asymptotic consensus if the network is $(2F+1)$-robust.
\label{thm:local_sufficient}
\end{theorem}

\subsection{Contagion and Bootstrap Percolation}
\label{sec:cascading}

Consider another class of diffusion dynamics, where a subset of nodes wish to spread their `status' through the whole network. 
Specifically, assume that each node in the network can be in one of two states: infected (e.g., with an idea or innovation) or uninfected. Starting with an initial set of infected nodes (which can be chosen deterministically or randomly), the infection spreads (or cascades) in discrete steps according to the following rule: each node becomes infected if at least $r$ of its neighbors have been infected and once infected, stays infected forevermore.  Here, $r$ is called the {\it threshold} for cascading. 
This class of dynamics appears in the study of contagion \cite{Easley10}, bootstrap percolation \cite{janson2012bootstrap,balogh2012sharp}, best response dynamics in strategic complements games \cite{Jackson08}, and resilient broadcasting \cite{ZhangSundaram2012}.

For some $m<n$ (where $n$ is the number of nodes in the network), we say there is {\it contagion from any $m$ nodes} \cite{Jackson08} (or equivalently, {\it bootstrap percolation succeeds from any $m$ nodes}) if {\it any} set of initially infected nodes with size $m$ causes the whole network to be eventually infected.
Note that for cascading with threshold $r$, $m\ge r$ is necessary to facilitate contagion; otherwise, no nodes will be infected besides the nodes in the initial set.
Further note that the minimum degree of the network is a fundamental limitation for the emergence of contagion, i.e., if the minimum degree is less than $r$, there always exists an initial set of size $m<n$ such that cascading with threshold $r$ fails to cause contagion. However, minimum degree by itself is not sufficiently useful to capture these dynamics. Consider again the network in Figure~\ref{fig:Counterexample}. Even though the network has large minimum degree, we can choose an initial set with $\frac{n}{2}$ nodes (either $\mathcal{S}_1$ or $\mathcal{S}_2$) such that cascading with any threshold bigger than 1 fails to produce contagion.
The following result from \cite{Easley10, Jackson08} (cast in the language of reachable sets) provides the condition for contagion to succeed.\footnote{The work in \cite{Easley10, Jackson08} considers a slightly different scenario where the cascading threshold is based on the {\it fraction} of a node's neighbors that are infected, but the extension to the absolute threshold case is trivial.}

\begin{theorem}
\label{thm:cascading}
For cascading with threshold $r$, contagion from any $m$ nodes occurs if and only if every subset of $\mathcal{V}$ with size up to $n-m$ is $r$-reachable, where $m\ge r$.
\end{theorem}

Given the fundamental role of the notions of reachable sets and robustness in the applications discussed above, we will start by studying the robustness of three common random graph models: Erd\H os-R\'enyi graphs, geometric random graphs, and preferential attachment graphs.  After that, we will analyze the computational complexity of determining the extent to which a given graph is robust.

\section{Robustness of Erd\H os-R\'enyi Random Graphs}
\label{sec:E-R}
Erd\H os-R\'enyi random graphs \cite{ErdosRenyi1961,Bollobas01} are one of the most common models for large-scale networks. 
The version we study here is denoted as $\mathcal{G}_{n,p}$: it consists of $n$ nodes and each possible (undirected) edge between two nodes is present independently with probability $p$ (which may be a function of $n$), and absent with probability $q =1-p$. 
Let the probability of an event be denoted by $\mathbb{P}(\cdot)$.  Recall that a {\it graph property} can be regarded as a class of graphs that is closed under isomorphism. A key feature of the $\mathcal{G}_{n,p}$ model is that we can explore properties that are shared by {\it almost all} graphs, a notion that is defined as follows.

\begin{definition}
Assume $\mathcal{P}$ is a graph property and $p=p(n)$ is a function of $n$. We say that \textbf{almost all} $G \in \mathcal{G}_{n,p}$ have property $\mathcal{P}$ if $\mathbb{P}(\mathcal{G}_{n,p}\in \mathcal{P}) \to 1$ as $n \to \infty$, and \textbf{almost no} $G \in \mathcal{G}_{n,p}$ has property $\mathcal{P}$ if $\mathbb{P}(\mathcal{G}_{n,p}\in \mathcal{P}) \to 0$ as $n \to \infty$.
\end{definition}

One important feature of  $\mathcal{G}_{n,p}$, which was demonstrated by Erd\H os and R\'enyi, is that the model exhibits  {\it phase transitions}.  More precisely, we define a {\it threshold function} as follows.

\begin{definition}[Threshold Function]
\label{def:Threshold_function}
Consider a function $t(n)$ of the form $\frac{g(n)}{n}$ where $g(n)\to\infty$ as $n\to\infty$, and a function $x=o(g(n))$ satisfying $x\to\infty$ as $n\to\infty$. We say $t(n)$ is a \textbf{threshold function} for a graph property $\mathcal{P}$ if $p(n)= \frac{g(n)+x}{n}$ implies that almost all $G \in \mathcal{G}_{n,p}$ have property $\mathcal{P}$ and $p(n)= \frac{g(n)-x}{n}$ implies that almost no $G \in \mathcal{G}_{n,p}$ has property $\mathcal{P}$. 
\end{definition}

All of the properties we are going to study in this section have threshold functions of the above form.\footnote{There are also various other definitions of threshold functions, but they are typically more conservative than the one we consider here \cite{Bollobas01}.} 
Loosely speaking, if the probability of adding an edge is `larger' than $t(n)$ in the sense indicated by Definition~\ref{def:Threshold_function}, then almost all $G \in \mathcal{G}_{n,p}$ will have property $\mathcal{P}$, and if the probability is `smaller' than $t(n)$, almost no $G \in \mathcal{G}_{n,p}$ will have this property.

\begin{definition}
For $G\in \mathcal{G}_{n,p}$ and constant $r\in \mathbb{Z}_{\geq1}$, define the properties of \textbf{being $r$-connected}, \textbf{$r$-robust} and \textbf{having minimum degree $r$} by \textbf{$K_r$}, \textbf{$R_r$} and \textbf{$D_r$}, respectively.
\end{definition}

\begin{lemma}[\cite{ErdosRenyi1961}]
For any constant $r\in \mathbb{Z}_{\geq1}$, $t(n)= \frac{\ln n + (r-1)\ln\ln n}{n}$ is a threshold function for property $K_r$. It is also a threshold function for property $D_r$.
\label{lem:threshold_connectivity}
\end{lemma}

The following is one of our main results: it establishes that the above threshold function for $r$-connectivity (and minimum degree $r$) is {\it also} a threshold function for $r$-robustness in Erd\H os-R\'enyi random graphs.

\begin{theorem}
For any constant $r\in \mathbb{Z}_{\geq1}$, $t(n)= \frac{\ln n + (r-1)\ln\ln n}{n}$ is a threshold function for property $R_r$.
\label{thm:threshold_robust}
\end{theorem}

From Lemmas~\ref{lem:threshold_connectivity} and \ref{lem:1_robust_connectivity}, the above theorem is immediately true for $r = 1$ since $1$-connectedness and $1$-robustness are equivalent.  To prove the theorem for $r \ge 2$, we will first need the following lemma showing that all subsets of nodes up to a certain size will be $r$-reachable when the probability is above the given threshold.

\begin{lemma}
Let $\alpha=\alpha(n)$ be a positive function satisfying $\sup_n{\alpha(n)} < 1$ and $\ln{\ln n}=o(\alpha \ln{n})$.   For any constant $r \in \mathbb{Z}_{\geq1}$, let $S_{r}^{1-\alpha}$ be the property that every subset of $\mathcal{V}$ with size up to $\lfloor(1-\alpha)n\rfloor$ is $r$-reachable. If $p(n)= \frac{\ln n + (r-1)\ln\ln n + x}{n}$, where $x=x(n)$ is some function satisfying $x=o(\ln{\ln{n}})$ and $x\to\infty$ as $n\to\infty$, then almost all $G \in \mathcal{G}(n,p)$ have property $S_{r}^{1-\alpha}$.
\label{lem:reachable_size}
\end{lemma}

\begin{proof}
Let $\mathbb{P}_{0}$ be the probability that some set of cardinality up to $n_c=\lfloor(1-\alpha)n\rfloor$ is not $r$-reachable. We need to prove that $\mathbb{P}_0=o(1)$ when $p(n)= \frac{\ln n + (r-1)\ln\ln n + x}{n}$.
Denote the probability that some set $\mathcal{S}\subset\mathcal{V}$ with cardinality $k$ (i.e., $\rvert \mathcal{S} \rvert =k$) is not $r$-reachable as $\mathbb{P}_k$.  By the union bound, we know that $\mathbb{P}_0 \le \sum_{k=1}^{n_c}\mathbb{P}_k$. 
For fixed $\mathcal{S}$ of cardinality $k$, the probability that a node $v\in \mathcal{S}$ has less than $r$ neighbors outside is $\sum_{i=0}^{r-1}\binom{n-k}{i}q^{n-k-i}p^{i}$, and the probability that $\mathcal{S}$ is not $r$-reachable is $(\sum_{i=0}^{r-1}\binom{n-k}{i}q^{n-k-i}p^{i})^k$, where $q=1-p$. 
Since there are $\binom{n}{k}$ such sets $\mathcal{S}$, we know that $\mathbb{P}_k\le \binom{n}{k}(\sum_{i=0}^{r-1}\binom{n-k}{i}q^{n-k-i}p^{i})^k$.
In the rest of the proof, we focus on the cases where $k\le n_c$.
Using the fact that $\binom{n}{k}\le (\frac{en}{k})^k$ and $\binom{n}{k} \le n^k$, we obtain the following upper bound for $\mathbb{P}_k$:
\begin{align}
\mathbb{P}_k
&\le \binom{n}{k}\left(\sum_{i=0}^{r-1}\binom{n-k}{i}q^{n-k-i}p^{i}\right)^k \nonumber \\
&\le \left(\frac{en}{k}\sum_{i=0}^{r-1}(np)^{i}(1-p)^{n-k-i}\right)^k  \nonumber \\
&\le \left(\frac{en}{k} (1-p)^{n-k} r\left(\frac{np}{1-p}\right)^{r-1} \right)^k \nonumber \\
&= \left(\frac{er}{(1-p)^{r-1}} \frac{n(np)^{r-1}(1-p)^{n}}{k(1-p)^k }\right)^k \nonumber \\
&\le \left(\frac{c_1 n(np)^{r-1}(1-p)^{n}}{k(1-p)^k }\right)^k \nonumber .
\end{align}
In the last step above, $c_1$ is some constant upper bound for $\frac{er}{(1-p)^{r-1}}$ satisfying $0<c_1<2er$ for {\it sufficiently large} $n$. By saying that a property holds for sufficiently large $n$, we mean that there exists some $n_0\in \mathbb{Z}_{\ge1}$ such that this property holds for all $n>n_0$. The notion of ``for sufficiently large $n$'' will be implicitly used throughout the proof. 
Note that $1-p\le e^{-p}$ and recall that $p(n)= \frac{\ln n + (r-1)\ln\ln n + x}{n}$. Thus,
\begin{align}
\mathbb{P}_k
&\le \left( \frac{c_1n(np)^{r-1}e^{-\ln{n}-(r-1)\ln{\ln{n}}-x}} {k(1-p)^k}\right)^k \nonumber \\
&= \left( \frac{c_1n(np)^{r-1}e^{-x}}{k(1-p)^kn(\ln{n})^{r-1}}\right)^k \nonumber \\
&= \left( c_1 \left(\frac{\ln{n}+(r-1)\ln{\ln{n}+x}}{\ln{n}}\right)^{r-1} \frac{e^{-x}}{k(1-p)^k}\right)^k \nonumber \\
&\le \left( \frac{c_2 e^{-x}}{k(1-p)^k}\right)^k \nonumber .
\end{align}
Note that $\frac{\ln{n}+(r-1)\ln{\ln{n}+x}}{\ln{n}}<2$ for sufficiently large $n$ and thus $0<c_2<c_1 2^{r-1}$. Next, note that $\ln{(1-p)}=-\sum_{i=1}^{\infty}\frac{p^i}{i}$ for $p\in [0,1)$, and thus,
\begin{align}
\mathbb{P}_k
&\le \left(\frac{c_2 e^{-x}}{k}\exp\{k\sum_{i=1}^{\infty}\frac{p^i}{i}\}\right)^k \nonumber \\
&= \left(\frac{c_2 e^{-x}}{k}\exp\{kp+ kp^2\sum_{i=2}^{\infty}\frac{p^{i-2}}{i}\}\right)^k \label{lem_first_appro} \\
&\le \left(c_2 e^{c_3}\frac{e^{-x}e^{kp}}{k}\right)^k = \left(\frac{c_4 e^{-x}e^{kp}}{k}\right)^k \nonumber .
\end{align}
Note that in (\ref{lem_first_appro}), since $\sum_{i=2}^{\infty}\frac{p^{i-2}}{i}< \sum_{i=2}^{\infty} p^{i-2} =\frac{1}{1-p}$ and $kp^2<np^2 = o(1)$, $0<c_3<1$ for sufficiently large $n$. Further note that $c_4=c_2 e^{c_3}$ and thus $0<c_4<re^2 2^r$.

Let $f(k)=\frac{e^{kp}}{k}$ be a function of $k$, where $k \in \mathbb{R}_{>0}$. Then $\frac{df}{dk}=\frac{e^{kp}(kp-1)}{k^2}$.   Since $\frac{df}{dk}<0$ if $k<\frac{1}{p}$ and $\frac{df}{dk} >0$ if $k> \frac{1}{p}$, $f(k)\le \max\{ f(1), f(n_c)\}$ for $k \in \{1, 2, \ldots, n_c\}$.
We know that $f(n_c)=\frac{\exp\{n_c p\}}{n_c}\le \frac{\exp\{(1-\alpha){np}\}}{(1-\alpha)n} = \frac{1}{1-\alpha}\exp\{ (1-\alpha)np - \ln{n} \} = \frac{1}{1-\alpha}\exp\{ -\alpha\ln{n} + (1-\alpha)(r-1)\ln{\ln{n}} +(1-\alpha)x \}$. Since $\alpha(n)$ is positive, strictly bounded below $1$ and $\ln{\ln n}=o(\alpha \ln{n})$, we know that $f(n_c)=o(1)$. Further note that $f(1)=e^p>1$. 
Thus, for sufficiently large $n$, $f(k)\le f(1) <e$ and $\mathbb{P}_k\le \big(c_4 e^{1-x}\big)^k$.
We now have
\begin{align}
\mathbb{P}_0
&\le \hspace{0.1cm} \sum_{k=1}^{n_c}\mathbb{P}_k \le \sum_{k=1}^{\infty} \big(c_4 e^{1-x}\big)^k = \frac{c_4 e^{1-x}}{1- c_4 e^{1-x}} = o(1), \nonumber
\end{align}
since $x\to\infty$ as $n\to\infty$, completing the proof.
\end{proof}

This lemma immediately leads to a proof of Theorem~\ref{thm:threshold_robust}.


\begin{IEEEproof}[Proof of Theorem~\ref{thm:threshold_robust}]
For the first part of the proof,  we show that for any constant $r\in \mathbb{Z}_{\ge1}$, if $p(n)= \frac{\ln n + (r-1)\ln\ln n + x}{n}$, where $x=x(n)$ is some function satisfying $x=o(\ln{\ln{n}})$ and $x\to\infty$ as $n\to\infty$, then almost all $G \in \mathcal{G}(n,p)$ are $r$-robust.  Applying Lemma~\ref{lem:reachable_size} with $\alpha = \frac{1}{2}$, we immediately see that in almost all $G\in\mathcal{G}_{n,p}$, any set of nodes with size up to $\lfloor \frac{n}{2} \rfloor$ will be $r$-reachable.  Thus, for almost all $G \in \mathcal{G}_{n,p}$, given any two disjoint and nonempty subsets of nodes of $G$, at least one of them will be $r$-reachable, and thus $G$ will be $r$-robust.

For the second part of the proof, we need to show that for any constant $r\in  \mathbb{Z}_{\ge1}$, if $p(n)= \frac{\ln n + (r-1)\ln\ln n - x}{n}$, where $x=x(n)$ is some function satisfying $x=o(\ln{\ln{n}})$ and $x\to\infty$ as $n\to\infty$, then almost no $G \in \mathcal{G}(n,p)$ is $r$-robust. The result is obtained by combining Lemma~\ref{lem:robust_connectivity} and Lemma~\ref{lem:threshold_connectivity}.
\end{IEEEproof}

\begin{remark}
The above theorem shows that Erd\H os-R\'enyi graphs gain a great deal more structure at the threshold $t(n)= \frac{\ln n + (r-1)\ln\ln n}{n}$ than simply being $r$-connected (or having minimum degree $r$). As argued earlier, whereas $r$-connectedness implies that given any two disjoint and nonempty sets, the nodes in at least one of the sets collectively have $r$ neighbors outside, the above result shows that there is (at least) one node in one of the sets that {\it by itself} has $r$ neighbors outside.  As an aside, note that the somewhat direct proof of $r$-robustness given above immediately yields a proof of $r$-connectedness of Erd\H os-R\'enyi graphs for connection probabilities above the given threshold.
\end{remark}

\subsection{Implications for Consensus and Contagion}
We now show what the above result means for resilient asymptotic consensus and the emergence of contagion in the $\mathcal{G}_{n,p}$ model.

\begin{definition}
For $G\in \mathcal{G}_{n,p}$ and constant $F\in \mathbb{Z}_{\geq1}$, define \textbf{$RAC_F$} to be the property that resilient asymptotic consensus is reached under the $F$-local model using the W-MSR algorithm for any initial values.
\end{definition}

\begin{corollary}
For any constant $F\in \mathbb{Z}_{\geq1}$, $t(n)= \frac{\ln n + 2F\ln\ln n}{n}$ is a threshold function for property $RAC_F$.
\label{cor:threshold_resilient_consensus}
\end{corollary}
\begin{IEEEproof}
As discussed in Section~\ref{subsec:conn} and Theorem~\ref{thm:local_sufficient}, $(2F+1)$-connectedness is necessary and $(2F+1)$-robust is sufficient, respectively, for the W-MSR algorithm to achieve resilient asymptotic consensus under the $F$-local model. Thus, by Lemma~\ref{lem:threshold_connectivity} and Theorem~\ref{thm:threshold_robust}, the result follows.
\end{IEEEproof}

\begin{definition}
For $G\in \mathcal{G}_{n,p}$, constant $r\in \mathbb{Z}_{\geq1}$, and positive function $\alpha=\alpha(n)$ satisfying $\sup_n{\alpha(n)} < 1$ and $\ln{\ln n}=o(\alpha \ln{n})$, define $C_{r}^{\alpha}$ to be the property that contagion from any $\lceil \alpha n \rceil$ nodes occurs when cascading with threshold $r$.
\end{definition}

\begin{corollary}
For any constant $r\in \mathbb{Z}_{\geq1}$ and positive function $\alpha=\alpha(n)$ satisfying $\sup_n{\alpha(n)} < 1$ and $\ln{\ln n}=o(\alpha \ln{n})$, $t(n)= \frac{\ln n + (r-1)\ln\ln n}{n}$ is a threshold function for property $C_{r}^{\alpha}$.
\label{cor:threshold_cascade}
\end{corollary}
\begin{IEEEproof}
Note that by Theorem~\ref{thm:cascading}, the properties $S_{r}^{1-\alpha}$ (defined in Lemma~\ref{lem:reachable_size}) and $C_{r}^{\alpha}$ are equivalent. Thus, the result follows by combining Lemma~\ref{lem:threshold_connectivity}, Lemma~\ref{lem:reachable_size} and the discussions in Section~\ref{sec:cascading}.
\end{IEEEproof}

The above corollary indicates that at the threshold $t(n)= \frac{\ln n + (r-1)\ln\ln n}{n}$, Erd\H os-R\'enyi graphs gain the ability to allow information initially held by any $\alpha$-fraction of nodes to cascade through the network to all other nodes.
The fraction $\alpha$ can go to $0$ at a sufficiently slow rate; for example, $\alpha$ can be some function in $\Omega(\frac{1}{(\ln n)^{\epsilon}})$, $0 < \epsilon < 1$, that satisfies $\sup_n{\alpha(n)} < 1$.
Theorem~\ref{thm:threshold_robust} (together with Corollaries \ref{cor:threshold_resilient_consensus} and \ref{cor:threshold_cascade}) implies that the `worst-case' networks (such as in Figure~\ref{fig:Counterexample}) will not appear (with probability tending to $1$) in Erd\H os-R\'enyi graphs.

\begin{remark}
Note that Corollary~\ref{cor:threshold_cascade} also applies to bootstrap percolation (due to the identical dynamics under the two scenarios). In fact, more general results for bootstrap percolation in  Erd\H os-R\'enyi graphs can be found in \cite{janson2012bootstrap}, where probability functions of the form $t(n)=\frac{g(n)}{n}$ are also explored. Our approach explicitly highlights the underlying graph properties of reachability and robustness and shows their relationship to connectivity and minimum degree in Erd\H os-R\'enyi graphs, leading to a fairly direct proof of the phase transition at the given threshold,  with corresponding implications for consensus and contagion (or bootstrap percolation).
\end{remark}

\section{Robustness of Geometric Random Graphs}
\label{sec:geometric}

Another widely used model for large networks is the {\it geometric random graph}, which captures edges between nodes that are in close (spatial) proximity to each other. We consider the geometric graph $\mathcal{G}_{n,\rho,l}^{d} = \{\mathcal{V}, \mathcal{E}\}$, which is an undirected graph generated by first placing $n$ nodes (according to some mechanism) in a region $\Omega_d=[0,l]^d$, where $d\in \mathbb{Z}_{\geq1}$.  We denote the position of node $i \in \mathcal{V}$ with $x(i) \in \Omega_d$.  Nodes $i, j \in \mathcal{V}$ are connected by an edge if and only if $\|x(i)-x(j)\| \le \rho$ for some threshold $\rho$, where $\|\cdot\|$ indicates an appropriate norm (often taken to be the standard Euclidean norm).
When the node positions are generated randomly (e.g., uniformly and independently) in the region, one obtains a {\it geometric random graph}.  In the widely-studied model $\mathcal{G}_{n,\rho}^{d}$, the parameter $l$ is fixed and graph properties are typically explored when $n\to\infty$ and $\rho\to 0$, leading to dense random networks \cite{Penrose03}.\footnote{Note that properties of networks with a finite number of nodes have also been explored (e.g., see \cite{Desai2002Connectivity}).}
In the more general model $\mathcal{G}_{n,\rho,l}^{d}$, however, the length $l$ is also allowed to increase and the density $\frac{n}{l^d}$ can converge to some constant, making it suitable for capturing both dense and sparse random networks \cite{Paolo2003}.

In Section~\ref{sec:E-R}, we showed that the properties of connectivity and robustness have the same threshold function in Erd\H os-R\'enyi graphs.  In this section, we will prove similar results for one-dimensional geometric random graphs (i.e., $d=1$).
We start by providing a result showing that connectivity and robustness cannot be very different in one-dimensional geometric graphs, and are in fact equal when the nodes are sufficiently spread out (regardless of how the node positions are generated and the relationships between $\rho$, $n$ and $\l$).  In the following, we assume that the nodes are ordered such that if $i, j \in \mathcal{V}$ and $i < j$ then $x(i) \le x(j)$.

\begin{theorem}
\label{thm:1-D}
In $\Omega_1=[0,l]$, if $\mathcal{G}_{n,\rho,l}^{1}$ is $r$-connected, then it is at least $\lfloor\frac{r}{2}\rfloor$-robust.  Furthermore, if $x(n)-x(1) > 3\rho$, then the graph is $r$-connected if and only if it is $r$-robust.
\end{theorem}

\begin{IEEEproof}
First, note that if $x(n)-x(1) \le \rho$ then the graph is complete and therefore $(n-1)$-connected and $\lceil\frac{n}{2}\rceil$-robust, and thus the claim holds.    In the rest of the proof, we assume that $x(n)-x(1) > \rho$.  In this case, if the graph is $r$-connected, the following two properties hold.
\begin{enumerate}
\item Every interval of the form $(a,a+\rho] \subset (x(1),x(n))$ must have at least $r$ nodes, because otherwise, removing the nodes in that interval would disconnect the nodes in the interval $[x(1),a]$ from those in the interval $(a+\rho,x(n)]$.  The same is true for every interval of the form $[a,a+\rho) \subset (x(1),x(n))$, and thus for every closed interval of length $\rho$ contained in $(x(1),x(n))$.
\item Consider any nonempty set $\mathcal{S} \subset \mathcal{V}$.  If there exists an interval $[a, a+\rho] \subset (x(1),x(n))$ with no nodes from $\mathcal{S}$, then there must be a node from $\mathcal{S}$ in the interval $[x(1),a)$ or in the interval $(a+\rho,x(n)]$.  By symmetry, assume that $\mathcal{S}$ has nodes in $[x(1),a)$ and let $i$ be the node in $\mathcal{S}$ that is closest to $a$ from this interval.  Then the interval $(x(i),x(i)+\rho]$ contains no nodes from $\mathcal{S}$, but contains at least $r$ nodes, and thus $\mathcal{S}$ is $r$-reachable.
\end{enumerate}

Now consider any two disjoint and nonempty subsets $\mathcal{S}_1, \mathcal{S}_2 \subset \mathcal{V}$, and any interval $[a, a+\rho] \subset (x(1),x(n))$.  If $\mathcal{S}_1$ (resp. $\mathcal{S}_2$) has no nodes in $[a, a+\rho]$, then $\mathcal{S}_1$ (resp. $\mathcal{S}_2$) is $r$-reachable.  Thus, suppose both $\mathcal{S}_1$ and $\mathcal{S}_2$ have nodes in $[a, a+\rho]$.  If $\mathcal{S}_1$ is not $\lfloor\frac{r}{2}\rfloor$-reachable, there are fewer than $\lfloor\frac{r}{2}\rfloor$ nodes from $\mathcal{S}_2$ in  $[a, a+\rho]$.  Choose any node $i$ from $\mathcal{S}_2$ in the interval.  There are at least $r-1$ remaining nodes in the interval, and at most $\lfloor\frac{r}{2}\rfloor - 1$ of them are in $\mathcal{S}_2$.  Thus $i$ has at least $r-1-\lfloor\frac{r}{2}\rfloor + 1 \ge \lfloor\frac{r}{2}\rfloor$ neighbors in the interval that are not in $\mathcal{S}_2$.  Therefore, for any two disjoint and nonempty subsets $\mathcal{S}_1, \mathcal{S}_2 \subset \mathcal{V}$, at least one of them is $\lfloor\frac{r}{2}\rfloor$-reachable.  Thus the graph is at least $\lfloor\frac{r}{2}\rfloor$-robust, proving the first part of the theorem.

For the second part of the theorem, assume $x(n)-x(1) > 3\rho$.  Then there exists an interval $[a,a+3\rho] \subset (x(1),x(n))$.  Consider any two nonempty and disjoint subsets $\mathcal{S}_1, \mathcal{S}_2 \subset \mathcal{V}$, and denote $\mathcal{X} = \mathcal{V} \setminus \left(\mathcal{S}_1 \cup \mathcal{S}_2\right)$.  By the argument above, if either $\mathcal{S}_1$ or $\mathcal{S}_2$ does not have any nodes in some closed interval of length $\rho$ within $(x(1),x(n))$, that set will be $r$-reachable.  Thus, suppose that both $\mathcal{S}_1$ and $\mathcal{S}_2$ have nodes in all closed intervals of length $\rho$ within $(x(1),x(n))$.  Pick any node $i$ from $\mathcal{S}_1$ in the interval $[a+\rho,a+2\rho]$, and let $j \in \mathcal{S}_2$ be the node in $[a+\rho,a+2\rho]$ that is closest to $i$.  We assume without loss of generality that $x(j) \le x(i)$ and that if $x(j) < x(i)$, then there are only nodes from $\mathcal{X}$ between $i$ and $j$ (the latter can always be enforced by redefining $i$ to be the node in $\mathcal{S}_1$ that is closest to $j$ in $[a+\rho,a+2\rho]$).

Suppose that $\mathcal{S}_1$ is not $r$-reachable.  Then there are fewer than $r$ nodes from $\mathcal{S}_2 \cup \mathcal{X}$ in the interval $[x(i)-\rho,x(i)+\rho]$.  If $x(j) = x(i)$, then $j$ has at least $2r$ neighbors in $[x(i)-\rho,x(i)+\rho]$ and since at most $r$ of them are from $\mathcal{S}_2$, the set $\mathcal{S}_2$ will be $r$-reachable.  Thus assume that $x(j) < x(i)$.  Let $n_{\mathcal{S}_2}$ and $n_{\mathcal{X}}$ be the number of nodes in $[x(i)-\rho,x(i))$ from $\mathcal{S}_2$ and $\mathcal{X}$, respectively.
Then the combined number of nodes from $\mathcal{S}_2$ and $\mathcal{X}$ in $[x(i),x(i)+\rho]$ must be less than $r - n_{\mathcal{S}_2} - n_{\mathcal{X}}$.  Let the number of nodes from $\mathcal{X}$ strictly between $j$ and $i$ be $n'_{\mathcal{X}}$, where $n'_{\mathcal{X}} \le n_{\mathcal{X}} < r$.  Since the interval $(x(j),x(j)+\rho]$ contains at least $r$ nodes, $x(j)$ has at least $r-n'_{\mathcal{X}}$ neighbors in the interval $[x(i),x(j)+\rho]$.  Since there are fewer than $r-n_{\mathcal{S}_2}-n_{\mathcal{X}}$ nodes from $\mathcal{S}_2$ in $[x(i),x(i)+\rho]$ (and thus in the smaller interval $[x(i),x(j)+\rho]$), node $j$ has at least $r-n'_{\mathcal{X}}-(r-n_{\mathcal{S}_2}-n_{\mathcal{X}}) = n_{\mathcal{S}_2}+n_{\mathcal{X}}-n'_{\mathcal{X}}$ neighbors outside its set in $[x(i),x(j)+\rho]$.  Furthermore, since there are at least $r-1$ nodes other than $j$ in $[x(i)-\rho,x(i))$, and $n_{\mathcal{S}_2}-1$ of them are from $\mathcal{S}_2$, node $j$ has at least $r-1-(n_{\mathcal{S}_2}-1) = r-n_{\mathcal{S}_2}$ neighbors outside its set in $[x(i)-\rho,x(i))$.  Thus, in the interval $[x(i)-\rho,x(j)+\rho]$, node $j$ has at least $r-n_{\mathcal{S}_2}+n_{\mathcal{S}_2}+n_{\mathcal{X}}-n'_{\mathcal{X}} \ge r$ nodes outside its set, making set $\mathcal{S}_2$ at least $r$-reachable.  Thus, the graph is $r$-robust.  The converse statement follows from Lemma~\ref{lem:robust_connectivity}
\end{IEEEproof}

Once again, note that the result in Theorem~\ref{thm:1-D} does not depend on {\it how} the positions of the nodes are generated.  Unfortunately, the theorem does not extend to geometric graphs in higher-dimensions.  For example, the graph shown in Figure~\ref{fig:Counterexample} can be viewed as a geometric graph in two dimensions, where the nodes in each set are all clustered horizontally within a distance $\rho$, and the two sets are vertically separated by a distance just below $\rho$ so that each node is within a distance $\rho$ of exactly one node in the opposite set.  Clearly that graph is only $1$-robust, despite having a connectivity of $\frac{n}{2}$.

Next we will present an asymptotic approach to analyzing one-dimensional random graphs (complementary to the analysis in Theorem~\ref{thm:1-D}).
We first define properties for {\it almost all} graphs in $\mathcal{G}_{n,\rho,l}^{d}$ as follows, similar to the $\mathcal{G}_{n,p}$ model.

\begin{definition}
Assume $\mathcal{P}$ is a graph property. We say that \textbf{almost all} $G \in \mathcal{G}_{n,\rho,l}^{d}$ have property $\mathcal{P}$ if $\mathbb{P}(\mathcal{G}_{n,\rho,l}^{d}\in \mathcal{P}) \to 1$ as $l \to \infty$, and \textbf{almost no} $G \in \mathcal{G}_{n,\rho,l}^{d}$ has property $\mathcal{P}$ if $\mathbb{P}(\mathcal{G}_{n,\rho,l}^{d}\in \mathcal{P}) \to 0$ as $l \to \infty$.
\end{definition}

Note that we study these properties in $\mathcal{G}_{n,\rho,l}^{d}$ as $l\to\infty$, and take $n$ and $\rho$ to be functions of $l$, i.e., $n=n(l)$ and $\rho=\rho(l)$.
We will use the following result from \cite{Paolo2003}.

\begin{theorem}[\cite{Paolo2003}]
\label{thm:1-D-Paolo}
Assume that $\rho n=kl\ln l$ for some constant $k>0$ and $\rho=\Omega(\frac{1}{\ln{l}})$.
\begin{itemize}
\item If $k>2$, or $k=2$ and $\rho\to \infty$, then almost all $G \in \mathcal{G}_{n,\rho,l}^{1}$ are connected.
\item If $\rho\in\Theta(l^{\epsilon})$ and $k\le (1-\epsilon)$ for some constant $0<\epsilon<1$, then almost no $G \in \mathcal{G}_{n,\rho,l}^{1}$ is connected.
\end{itemize}
\end{theorem}

We now present the following conditions under which the one-dimensional geometric random graph becomes $r$-connected and $r$-robust; the proof of this result builds upon and generalizes the proof of Theorem~\ref{thm:1-D-Paolo} from \cite{Paolo2003}.  Note that if $\rho(l) \ge l$, the graph will be $(n(l)-1)$-connected and $\lceil\frac{n(l)}{2}\rceil$-robust, and thus we focus on the case where $\rho(l) < l$ in the theorem below.

\begin{theorem}
\label{thm:1-Drandom}
Assume that $\rho n=kl\ln l$ for some constant $k>0$. 
\begin{itemize}
\item If $\rho < l$ and $\rho \in \Omega(l)$, then almost all $G \in \mathcal{G}_{n,\rho,l}^{1}$ are $r$-connected and $r$-robust for all $r \in \mathbb{Z}_{\ge{1}}$.
\item If $\rho = o(l) $ and $\rho l^{\frac{k}{r+1} -1} \to\infty$ for some $r \in \mathbb{Z}_{\ge{1}}$, then almost all $G \in \mathcal{G}_{n,\rho,l}^{1}$ are $r$-connected and $r$-robust.
\item If $\rho\in\Theta(l^{\epsilon})$ and $k\le (1-\epsilon)$ for some constant $0<\epsilon<1$, then almost no $G \in \mathcal{G}_{n,\rho,l}^{1}$ is $r$-connected or $r$-robust.
\end{itemize}
\end{theorem}
\begin{IEEEproof}
Fix any $r \in \mathbb{Z}_{\ge{1}}$.  In order to prove the first two parts, we will show that any interval of length $\rho$ contains at least $r$ nodes; the results will then follow from the arguments in the proof of Theorem~\ref{thm:1-D}. Let $\Omega_1=[0,l]$ be subdivided into non-overlapping segments of length $h=\frac{\rho}{r+1}$. Then $\Omega_1$ has $c=\lfloor \frac{(r+1)l}{\rho} \rfloor$ whole segments and potentially a fraction of a segment.  Any interval of length $\rho$ in $\Omega_1$ will contain at least $r$ whole segments and thus we just need to show every whole segment contains at least one node. 

Let $\omega$ be a random variable representing the number of empty whole segments. Since $\omega$ is a nonnegative integer random variable, by Markov's inequality we know $\mathbb{P}(\omega>0)\le\mathbb{E}(\omega)$, where $\mathbb{E}(\omega)=c(1- \frac{h}{l})^n$ is the expected value of $\omega$.  Since $1-x\le\exp(-x)$, we have
\begin{align}
\mathbb{E}(\omega)
&= c\left( 1- \frac{h}{l}\right)^n \le c \exp\left(- \frac{nh}{l}\right) \nonumber \\
&\le \frac{(r+1)l}{\rho}   \exp\left( - \frac{n\rho}{(r+1)l  } \right)  \label{equ:scaling_geo} \\
&= \frac{(r+1)l}{\rho} \exp\left( - \frac{k}{r+1} \ln{l} \right) \nonumber \\
&= \frac{(r+1)}{\rho}  l^{1- \frac{k}{r+1}} \nonumber .
\end{align}
Note that in~(\ref{equ:scaling_geo}), we replaced $n$ by $\frac{kl\ln l}{\rho}$.

Under the conditions in the first part of the theorem, $\rho l^{\frac{k}{r+1} -1} \to\infty$ regardless of the choice of $r \in \mathbb{Z}_{\ge{1}}$.  Thus $\mathbb{E}(\omega) \rightarrow 0$ and Theorem~\ref{thm:1-D} indicates that almost all graphs will be $\lfloor\frac{r}{2}\rfloor$-robust for all $r \in \mathbb{Z}_{\ge{1}}$ (or equivalently, $r$-robust for all $r \in \mathbb{Z}_{\ge{1}}$).  By Lemma~\ref{lem:robust_connectivity}, almost all graphs will be $r$-connected for all $r \in \mathbb{Z}_{\ge{1}}$.  Similarly, for the second part, $\mathbb{E}(\omega)\to 0$ as $l\to\infty$ if $k$ and $r$ satisfy the given conditions, indicating that the graph will be $r$-robust and $r$-connected  (again, using Theorem~\ref{thm:1-D}).

For the third part, Theorem~\ref{thm:1-D-Paolo} indicates that almost no $G \in\mathcal{G}_{n,\rho,l}^{d}$ is connected under the given conditions, and thus almost no graph is $r$-connected or $r$-robust for any $r \ge 1$.
\end{IEEEproof}



\section{Robustness of Preferential Attachment Networks}
\label{sec:preatt}

Before discussing the preferential attachment model for complex networks, we start by reviewing the following construction method for $r$-robust graphs from \cite{ZhangSundaram2012, leblanc2013resilient}.

\begin{theorem}[\cite{ZhangSundaram2012, leblanc2013resilient}]
Let $\mathcal{G}=\{\mathcal{V},\mathcal{E}\}$ be an $r$-robust graph. Then graph
$$
\mathcal{G}'=\{\{\mathcal{V},v_{\text{new}}\},\{\mathcal{E},\mathcal{E}_{\text{new}}\}\},
$$
where $v_{\text{new}}$ is a new node added to $\mathcal{G}$ and $\mathcal{E}_{\text{new}}$ is the edge set related to $v_{\text{new}}$, is $r$-robust if $d_{v_{\text{new}}} \ge r$.
\label{thm:design robust}
\end{theorem}

The above theorem indicates that to build an $r$-robust graph with $n$ nodes (where $n \ge r$), we can start with an $r$-robust graph of order less than $n$ (such as a complete graph), and continually add new nodes with incoming edges from at least $r$ nodes in the existing graph.
The theorem does not specify {\it which} existing nodes should be chosen as neighbors.  When the nodes are selected with a probability proportional to the number of edges that they already have, this is known as {\it preferential-attachment} and leads to the formation of so-called {\it scale-free} networks \cite{Albert02}.  Specifically, the construction process in Theorem~\ref{thm:design robust} coincides with the {\it Barab\'asi-Albert (BA) model} \cite{Albert02}: start with a network of $r_0$ nodes and add new nodes to the network one at a time, where each new node connects to $r$ existing nodes chosen by the preferential-attachment mechanism.

\begin{theorem}
In the BA model, when the initial network is $r$-robust, then the generated scale-free network is $r$-connected (and has minimum degree at least $r$) if and only if the network is $r$-robust. 
\end{theorem}
\begin{IEEEproof}
Note that in the BA model, if there exists some new node which connects to less than $r$ existing nodes, then the network will have minimum degree less than $r$, and so will be neither $r$-connected nor $r$-robust; on the other hand, if all of the new nodes connect to $r$ existing nodes, then by Theorem~\ref{thm:design robust}, the network will be $r$-robust and thus $r$-connected (and with minimum degree $r$).
\end{IEEEproof}

To the extent that the BA model is a plausible mechanism for the formation of complex networks, our analysis indicates that these networks will also facilitate dynamics such as resilient consensus, provided that $r$ is sufficiently large when the network is forming.

\section{Complexity of Determining Degree of Robustness in General Graphs}
\label{sec:complexity}

While we focused on random graphs in the previous sections, we now consider the problem of determining the extent to which any {\it given} graph is robust.  Recall that by Lemma~\ref{lem:1_robust_connectivity}, 1-robustness is equivalent to being 1-connected.
Since there exist polynomial time algorithms to determine graph connectivity~\cite{CLRS01}, 1-robustness can be checked in polynomial time. However, in the rest of this section we will show that finding whether general graphs are $r$-robust for $r\ge 2$ is $\coNP$-complete. This is done by showing that a problem closely related to the robustness problem is $\NP$-complete in general graphs. Before discussing the reduction we need to introduce the following concepts \cite{CLRS01}, and define the robustness and $r$-robustness problems formally. 

\begin{definition}[Cut and Cut-set]
\label{def:cut}
For a graph $\mathcal{G}=\{\mathcal{V},\mathcal{E}\}$, a \textbf{cut} $\mathcal{C}=(\mathcal{S},\mathcal{V}\backslash \mathcal{S})$ is defined as a partition of nodes of $\mathcal{G}$ into two nonempty subsets $\mathcal{S}\subset \mathcal{V}$ and $\mathcal{V}\backslash \mathcal{S}$. The \textbf{cut-set} of a cut $\mathcal{C}=(\mathcal{S},\mathcal{V}\backslash \mathcal{S})$ is defined as the subset of the edges of $\mathcal{G}$ with one endpoint in $\mathcal{S}$ and the other in $\mathcal{V}\backslash \mathcal{S}$.
\end{definition}

\begin{definition}[Robustness and $r$-Robustness Problems]
\label{def:r-robustness}
Given a graph $\mathcal{G}$, the \textbf{robustness problem} determines the largest value of $r$ such that $\mathcal{G}$ is $r$-robust. The \textbf{$r$-robustness problem} is the decision version of the robustness problem that determines whether graph $\mathcal{G}$ is $r$-robust for a given $r\in \mathbb{Z}_{\ge1}$.
\end{definition}

If a graph is not $r$-robust, then there exist two nonempty and disjoint subsets of nodes $\mathcal{A},\mathcal{B}$ such that all nodes in these sets have at most $r-1$ neighbors outside their containing sets. Note that nodes in set $\mathcal{X}=\mathcal{V}\backslash (\mathcal{A}\cup \mathcal{B})$ can have any number of neighbors outside $\mathcal{X}$.
There is no apparent way to certify that a graph is $r$-robust without checking all pairs of disjoint and nonempty subsets of nodes and showing that at least one set out of each pair is $r$-reachable.
This is intractable as the number of such subsets is exponential in the size of the input graph.
On the other hand, to certify that a graph is {\it not} $r$-robust, one only needs to provide a single pair of disjoint and nonempty subsets of nodes, of which neither set is $r$-reachable.
Therefore, in the $r$-robustness problem, the `No' instances (input graphs that are not $r$-robust) have certifications that can be checked in polynomial time, and so the $r$-robustness problem is in the complexity class $\coNP$~\cite{Arora'09}.

The complexity of the robustness problem is equivalent to the complexity of the $r$-robustness problem within a factor of $O(\log n)$ (as a binary search can be used for finding the largest value of $r$ for which the graph is $r$-robust). To prove the $\coNP$-completeness of the $r$-robustness problem, we instead show that the complement of the $r$-robustness problem, which we call the ``$\rho$-degree cut problem", is $\NP$-hard. Note that the complement of a decision problem is obtained by reversing the `Yes' and `No' answers of all input instances. In other words, if problems $\mathcal{P}_1$ and $\mathcal{P}_2$ are complements, then the output of $\mathcal{P}_1$ to an input instance is `Yes' if and only if the output of $\mathcal{P}_2$ to that instance is `No'. Therefore, the complement of a problem in $\NP$ is in $\coNP$, and vice versa. Hence, the $\rho$-degree cut problem, formally defined below, is in $\NP$.

\begin{definition}[$\rho$-Degree Cut]
\label{def:r-cut-deg}
Given a graph $\mathcal{G}=\{\mathcal{\mathcal{V}},\mathcal{E}\}$, a \textbf{$\rho$-degree cut} is a pair of nonempty and disjoint subsets of nodes $\mathcal{A},\mathcal{B}\subset \mathcal{V}$ such that each node in $\mathcal{A}$ (resp. $\mathcal{B}$) has at most $\rho$ neighbors outside $\mathcal{A}$ (resp. $\mathcal{B}$), where $\rho \in \mathbb{Z}_{\geq{0}}$.  The \textbf{$\rho$-degree cut problem} determines whether the graph has a $\rho$-degree cut.
\end{definition}

It can be shown that if a problem is $\NP$-hard, then its complement is $\coNP$-hard (see~\cite{Arora'09} for a detailed discussion on complexity classes $\NP$ and $\coNP$). Hence, instead of directly showing that the $r$-robustness problem is $\coNP$-complete, we show that its complement, i.e. the $\rho$-degree cut problem, is $\NP$-complete. Moreover, recall that the $r$-robustness problem is in $\coNP$; therefore, knowing that the $\rho$-degree cut problem is $\NP$-complete, it can be concluded that the $r$-robustness problem is $\coNP$-complete. To prove the $\NP$-completeness of the $\rho$-degree cut problem we first study the hardness of a relaxed version of this problem called the {\it relaxed-$\rho$-degree cut problem}, defined as follows.

\begin{definition}[Relaxed-$\rho$-Degree Cut]
\label{def:R-r-cut-deg}
A \textbf{relaxed-$\rho$-degree cut} in a given graph $\mathcal{G}$ is a cut $\mathcal{C}=(\mathcal{A},\mathcal{V}\backslash\mathcal{A})$ such that each node of the graph is incident to at most $\rho$ edges in the cut-set, where $\rho \in \mathbb{Z}_{\geq{0}}$.  The \textbf{relaxed-$\rho$-degree cut problem} determines whether there exists a relaxed-$\rho$-degree cut in the input graph.
\end{definition}

Note that the difference between a $\rho$-degree cut and a relaxed-$\rho$-degree cut is that in the latter, we are interested in finding two nonempty sets that are each at most $\rho$-reachable and that partition the nodes of the graph, whereas in the former the two sets simply have to be nonempty, at most $\rho$-reachable, and disjoint.  Both the $\rho$-degree cut problem and the relaxed-$\rho$-degree cut problem are in complexity class $\NP$, as they possess certifications for `Yes' instances that can be checked in polynomial time. 
The relaxed-1-degree cut problem is equivalent to a known problem called the ``matching-cut problem" in which the goal is to find whether there exists a cut in the graph that is also a matching, i.e., no two edges in the cut-set share an end-point~\cite{PatrignaniPizzonia'01}. 
In~\cite{PatrignaniPizzonia'01} it was shown that the matching-cut problem is $\NP$-complete via a reduction from NAE3SAT~\cite{Schaefer'78}. By the equivalence of the relaxed-1-degree cut problem and the matching-cut problem, this shows the $\NP$-completeness of the relaxed-1-degree cut problem as well. However, in order to prove the $\NP$-completeness of the 1-degree cut problem (and subsequently, the $\rho$-degree cut problem), we will need some alterations to the proof in~\cite{PatrignaniPizzonia'01}. We will thus start by modifying the reduction in~\cite{PatrignaniPizzonia'01} from NAE3SAT to the matching-cut problem (or relaxed-1-degree cut problem) and then extend that to prove the $\NP$-completeness of the 1-degree cut problem. We start by defining NAE3SAT formally~\cite{Schaefer'78}.

\begin{definition}[NAE3SAT]
\label{def:NAE3SAT}
For a set of clauses each containing three literals from a set of boolean variables in \emph{Conjunctive Normal Form} (CNF), \textbf{NAE3SAT} determines whether there exists a truth assignment of the variables so that each clause contains at least one `True' and one `False' literal.
\end{definition}

It was shown by Schaefer in~\cite{Schaefer'78} that NAE3SAT is $\NP$-complete. Here, we discuss a reduction from NAE3SAT to the relaxed-1-degree cut problem by constructing a graph $\mathcal{G}(\phi)$ for any given CNF formula $\phi$ such that $\mathcal{G}(\phi)$ has a relaxed-1-degree cut (i.e., $\mathcal{G}(\phi)$ has a cut where each node of the graph is incident to at most one edge in the cut-set) if and only if $\phi$ can be satisfied within the NAE3SAT constraints. 
Let formula $\phi$ consist of $m$ clauses $C_1,\ldots, C_m$, where each clause contains three literals from the set of variables $X=\{x_1,\ldots,x_t\}$. 
The construction of graph $\mathcal{G}(\phi)$ from a given CNF formula $\phi$ is as follows.

\begin{figure}[t]
\centering
\includegraphics[scale=0.4]{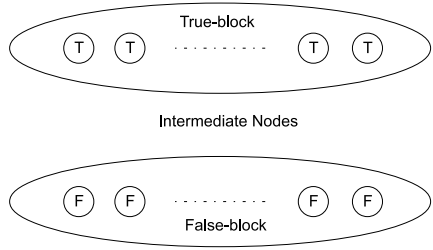}
\caption{The True-block and False-block in the construction of $\mathcal{G}(\phi)$. Each block is a complete subgraph with $4m+t$ nodes.}
\label{fig:TF-blocks}
\end{figure}

\begin{figure}[t]
\centering
\includegraphics[scale=0.45]{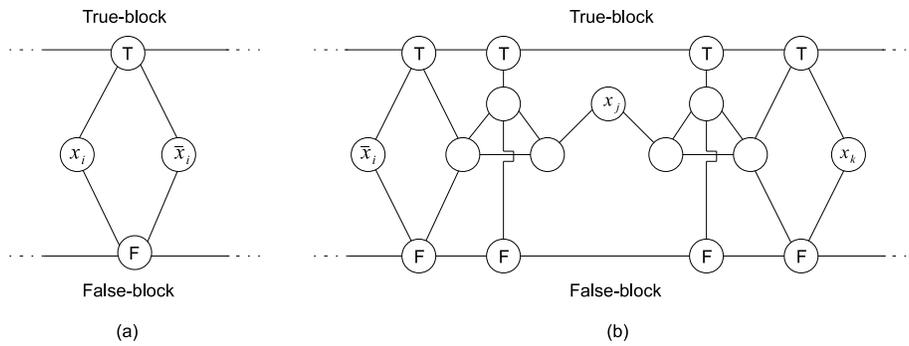}
\caption{Figure (a) demonstrates the variable-gadget for variable $x_i$, and (b) shows the clause-gadget for clause $\bar{x}_i \vee x_j \vee x_k$.}
\label{fig:var-clause-gadget}
\end{figure}

We first build two \emph{blocks}, where each block is a complete graph of $4m+t$ nodes. The upper and lower blocks are labeled the \emph{True-block} and \emph{False-block}, respectively, as illustrated in Figure~\ref{fig:TF-blocks}. We complete the construction of $\mathcal{G}(\phi)$ by adding subgraphs representing the variables and clauses of $\phi$ to these blocks in a carefully chosen way.  

\begin{figure}[t]
\centering
\includegraphics[scale=0.45]{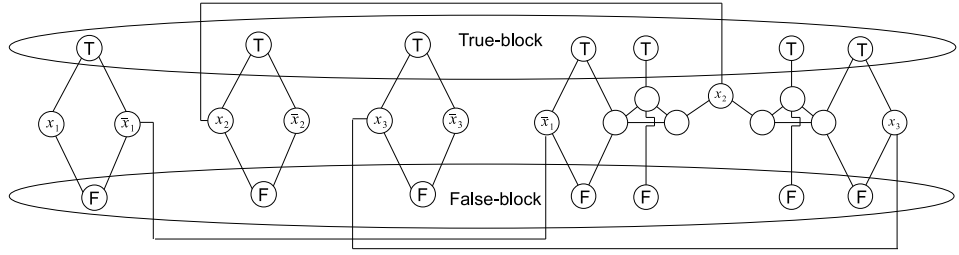}
\caption{The construction of clause $\bar{x}_1\vee x_2\vee x_3$ with each literal-node connected to its corresponding variable-node in the variable-gadgets.  The edges within the True and False-blocks have been omitted for clarity.}
\label{fig:var-clause}
\end{figure}

The subgraphs to be added to the blocks are of two types: (i) variable-gadgets, and (ii) clause-gadgets. A variable-gadget is incorporated for each variable $x_i\in X$. This gadget contains two nodes representing $x_i$ and $\bar{x}_i$ (the binary complement of $x_i$), each connected to the True and False-blocks as illustrated in Figure~\ref{fig:var-clause-gadget}-(a). Moreover, for each clause in $\phi$, a clause-gadget is constructed by connecting three nodes (each representing a literal of the clause) in addition to some extra nodes to the True and False-blocks as depicted in Figure~\ref{fig:var-clause-gadget}-(b). Finally, there are edges, called the {\it intermediate edges}, connecting each literal-node in each clause-gadget to its corresponding variable-node in the variable-gadgets. An example of $\mathcal{G}(\phi)$ for $\bar{x}_1\vee x_2\vee x_3$ is demonstrated in Figure~\ref{fig:var-clause}.

If graph $\mathcal{G}(\phi)$ has a relaxed-1-degree cut, then there exists a cut $\mathcal{C}=(\mathcal{A},\mathcal{V}\backslash\mathcal{A})$ that partitions the nodes of $\mathcal{G}(\phi)$ into two sets $\mathcal{A}$ and $\mathcal{B}=\mathcal{V}\backslash \mathcal{A}$ such that no node of the graph has more than one neighbor outside its set. In the following lemmas, we show that this cut $\mathcal{C}$ satisfies some useful properties (all of these lemmas assume that graph $\mathcal{G}(\phi)$ has a relaxed-1-degree cut and pertain to the cut $\mathcal{C}=(\mathcal{A},\mathcal{B})$ just described). Proofs of all subsequent lemmas and theorems in this section are given in the Appendix.

\begin{lemma}
\label{prop:cuts-TFblocks}
Let $\mathcal{T}$ (resp. $\mathcal{F}$) be the set of all nodes in the True-block (resp. False-block) of graph $\mathcal{G}(\phi)$. Then, $\mathcal{T}\subseteq \mathcal{A}$ or $\mathcal{T}\subseteq \mathcal{B}$ (resp. $\mathcal{F}\subseteq \mathcal{A}$ or $\mathcal{F}\subseteq \mathcal{B}$).
\end{lemma}

By the above lemma, cut $\mathcal{C}$ cannot go through the True and False-blocks of $\mathcal{G}(\phi)$. We assign `True' values to the nodes in the variable and clause-gadgets that are connected to the True-block by $\mathcal{C}$ and `False' values to the nodes that are connected to the False-block by $\mathcal{C}$.

\begin{lemma}
\label{prop:cuts-var-clause-gadget}
Cut $\mathcal{C}=(\mathcal{A},\mathcal{B})$ has the following two properties:
\begin{enumerate}
  \item For each variable-gadget, cut $\mathcal{C}$ leaves the variable-node and its negation node in opposite sets, i.e., they have opposite truth assignments.
  \item For each clause-gadget, cut $\mathcal{C}$ leaves at least one literal-node in set $\mathcal{A}$ and one literal-node in $\mathcal{B}$, i.e., at least one literal-node is assigned `True' and one is assigned `False.'
\end{enumerate}
\end{lemma}

\begin{lemma}
\label{prop:cuts-inter-edges}
All literal-nodes have the same truth values as their corresponding variable-nodes.
\end{lemma}

Using the properties stated in the above lemmas, we obtain the following result.
\begin{lemma}
\label{lem:r-r-r-hard}
The relaxed-1-degree cut problem is $\NP$-complete.
\end{lemma}

Recall that the relaxed-1-degree cut problem is equivalent to the matching-cut problem that was shown to be $\NP$-complete in~\cite{PatrignaniPizzonia'01}. The difference between the proof we provided here and the proof in~\cite{PatrignaniPizzonia'01} is in the construction of the clause-gadgets; our construction will allow us to show the $\NP$-completeness of the more general 1-degree cut problem as follows.

We first construct a graph $\mathcal{H}(\phi)$ by taking three copies of $\mathcal{G}(\phi)$ and adding edges to form one complete subgraph on all nodes in the three True-blocks and  another complete subgraph on all nodes in the three False-blocks. We refer to each of these copies of $\mathcal{G}(\phi)$ used in building $\mathcal{H}(\phi)$ as a {\it box}.
Figure~\ref{fig:H-var-clause} illustrates $\mathcal{H}(\phi)$ using the graph $\mathcal{G}(\phi)$ shown in Figure~\ref{fig:var-clause} for $\phi=\bar{x}_1\vee x_2\vee x_3$. Using this construction we can now prove the following result.

\begin{theorem}
\label{thm:robustness-hard}
The 1-degree cut problem is $\NP$-complete.
\end{theorem}

\begin{figure}[ht]
\centering
\includegraphics[scale=0.50]{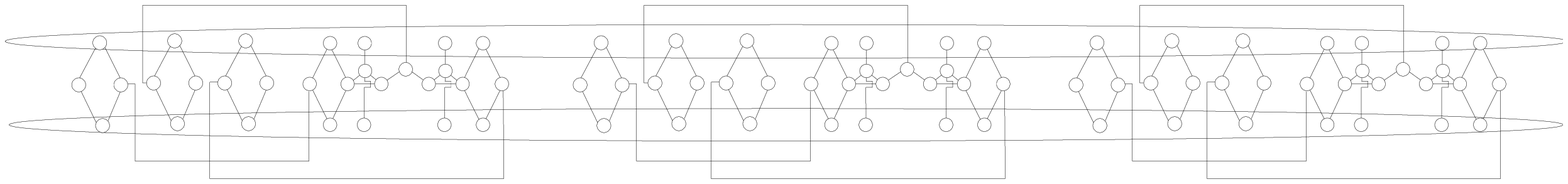}
\caption{The construction of $\mathcal{H}(\phi)$ from the graph $\mathcal{G}(\phi)$ depicted in Figure~\ref{fig:var-clause}.}
\label{fig:H-var-clause}
\end{figure}

The $\NP$-hardness of the 1-degree cut problem results in the $\coNP$-hardness of its complement problem, i.e., the 2-robustness problem. Since the 2-robustness problem is in $\coNP$, we can conclude that it is $\coNP$-complete. We now seek a stronger result, that is, showing that for any $r\in \mathbb{Z}_{\ge2}$, the $r$-robustness problem is $\coNP$-complete. This requires showing the $\NP$-completeness of the $\rho$-degree cut problem for any $\rho\in \mathbb{Z}_{\ge1}$, and to prove this, we first show that for any $\rho\in \mathbb{Z}_{\ge1}$, the relaxed-$\rho$-degree cut problem is $\NP$-complete. This is done in the following lemma by making certain modifications to the graph $\mathcal{G}(\phi)$ built for the NAE3SAT instance $\phi$ (e.g., Figure~\ref{fig:var-clause}).

\begin{lemma}
\label{lem:r-r-deg-cut}
For any $\rho\in \mathbb{Z}_{\ge1}$, the relaxed-$\rho$-degree cut problem is $\NP$-complete.
\end{lemma}

Using the above lemma, we can now prove the following stronger result.

\begin{theorem}
\label{thm:r-deg-cut}
For any $\rho\in \mathbb{Z}_{\ge1}$, the $\rho$-degree cut problem is $\NP$-complete.
\end{theorem}

Knowing that the $\rho$-degree cut is $\NP$-hard for any $\rho\in \mathbb{Z}_{\ge1}$, we conclude that its complement problem, i.e., the $r$-robustness problem for any $r\in \mathbb{Z}_{\ge2}$ is $\coNP$-hard. Combining this with the fact that the $r$-robustness problem is in $\coNP$ gives the following result. 

\begin{corollary}
\label{cor:robustness-hard}
For any $r \in \mathbb{Z}_{\ge{2}}$, the $r$-robustness problem is $\coNP$-complete.
\end{corollary}

Having established the complexity of the robustness problem, we now turn to the problem of \emph{approximating} the degree of robustness of graphs. 
To show that a graph is not $r$-robust, one needs to find an $(r-1)$-degree cut in that graph. 
For any given graph $\mathcal{G}$, let $\mathrm{OPT}(\mathcal{G})$ be the smallest nonnegative integer $\rho$ such that $\mathcal{G}$ has a $\rho$-degree cut.  Suppose that we have a (polynomial-time) approximation algorithm $\mathrm{ALG}$ whose input is graph $\mathcal{G}$ and whose output $\mathrm{ALG}(\mathcal{G})$ guarantees that graph $\mathcal{G}$ has an $\mathrm{ALG}(\mathcal{G})$-degree cut.
Define the {\it approximation ratio} $\alpha$ of the algorithm to be such that $\mathrm{ALG}(\mathcal{G}) \le \alpha\mathrm{OPT}(\mathcal{G})$ for all graphs $\mathcal{G}$.  Observe that since the $\rho$-degree cut problem is $\NP$-complete, one cannot hope to reach an approximation ratio of $\alpha=1$, unless $\Poly=\NP$.  In the following lemma, we show that it is unlikely to find an approximation algorithm for this problem with approximation ratio less than 2.

\begin{lemma}
\label{lem:hardness-apx}
The $\rho$-degree cut problem is not approximable within any factor less than 2, unless $\Poly=\NP$.
\end{lemma}


\section{Summary}
\label{sec:summary}
In this paper, we studied a graph property known as robustness which plays a key role in certain dynamics such as resilient consensus, contagion and bootstrap percolation.
While it is $\coNP$-complete to determine the degree of robustness in general graphs, and one can construct worst-case networks with very large connectivity (and minimum degree) and low robustness, we showed that the notions of robustness and connectivity coincide in three common models for complex networks.
In Erd\H os-R\'enyi random graphs, we showed that $r$-connectivity (and minimum degree) and $r$-robustness share the same threshold function. In one-dimensional geometric graphs, we proved that if the nodes are sufficiently spread apart,  $r$-connectedness is equivalent to $r$-robustness (regardless of how the node locations are generated). In the BA model for preferential attachment networks, we showed that when the initial network is robust, connectivity (and minimum degree) and robustness are equivalent. These findings indicate that those networks possess structure that makes them conducive to the dynamics described above; the implication of this for other classes of dynamics is a promising direction for future research.


\begin{appendix}
\subsection{Proof of Lemma~\ref{prop:cuts-TFblocks}}

\begin{IEEEproof}
Since each block is a complete graph with more than three nodes, cut $\mathcal{C}=(\mathcal{A},\mathcal{B})$ cannot separate the nodes in the same block; otherwise, there exists a node in the block that has at least two neighbors outside its own set.
\end{IEEEproof}

\subsection{Proof of Lemma~\ref{prop:cuts-var-clause-gadget}}
\begin{IEEEproof}
By Lemma~\ref{prop:cuts-TFblocks}, there are only two cases to consider: (i) all nodes in both the True and False-blocks are in $\mathcal{A}$ (or in $\mathcal{B}$), and (ii) all nodes in the True-block are in $\mathcal{A}$ and all nodes in the False-block are in $\mathcal{B}$ (or vice versa).

In case (i), if there exists a node from a variable-gadget in set $\mathcal{B}$, then that node immediately has at least two neighbors in $\mathcal{A}$, contradicting the definition of cut $\mathcal{C}$ (see Figure~\ref{fig:var-clause-gadget}-(a)). Similarly, it can be argued as follows that no node of any clause-gadget can be in set $\mathcal{B}$.  Referring to Figure~\ref{fig:5clause-gadget-v2}, the nodes labeled 1, 2, 3, 7, 8, and 9 cannot be in $\mathcal{B}$ since they would then have at least two neighbors in $\mathcal{A}$. Since nodes 2, 3, 7, and 8 are in $\mathcal{A}$, nodes 4 and 6 cannot be in $\mathcal{B}$ either. Then node 5 should also be in $\mathcal{A}$.  Hence, the only possibility is that all nodes in variable-gadgets and clause-gadgets are in $\mathcal{A}$. This makes $\mathcal{B}$ empty and violates the definition of cut $\mathcal{C}$. Thus, case (i) cannot hold and it only remains to study case (ii).

\begin{figure}[h]
\centering
\includegraphics[scale=0.5]{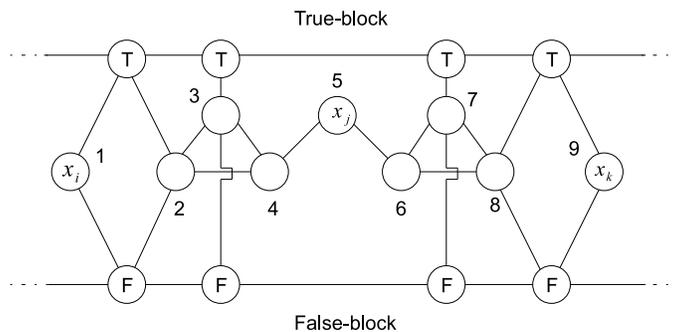}
\caption{All nine nodes in a clause-gadget with labels.}
\label{fig:5clause-gadget-v2}
\end{figure}

\begin{figure}[h]
\centering
\includegraphics[scale=0.5]{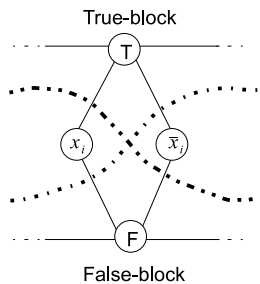}
\caption{The only two possible cuts through a variable-gadget resulting in all nodes having at most one neighbor on opposite sides of the cut.}
\label{fig:cut-var-gadget}
\end{figure}

In case (ii), if both nodes of a variable-gadget are in the same set (say $\mathcal{A}$), then a node from the False-block in set $\mathcal{B}$ has two neighbors in $\mathcal{A}$ (as seen in Figure~\ref{fig:var-clause-gadget}-(a)). This contradicts the definition of cut $\mathcal{C}$. The only possible cuts through variable-gadgets for this case are shown in Figure~\ref{fig:cut-var-gadget}, which leave any variable-node and its negation node on opposite sides of $\mathcal{C}$ and thus concludes the first property in the lemma.

\begin{figure}[h]
\centering
\includegraphics[scale=0.33]{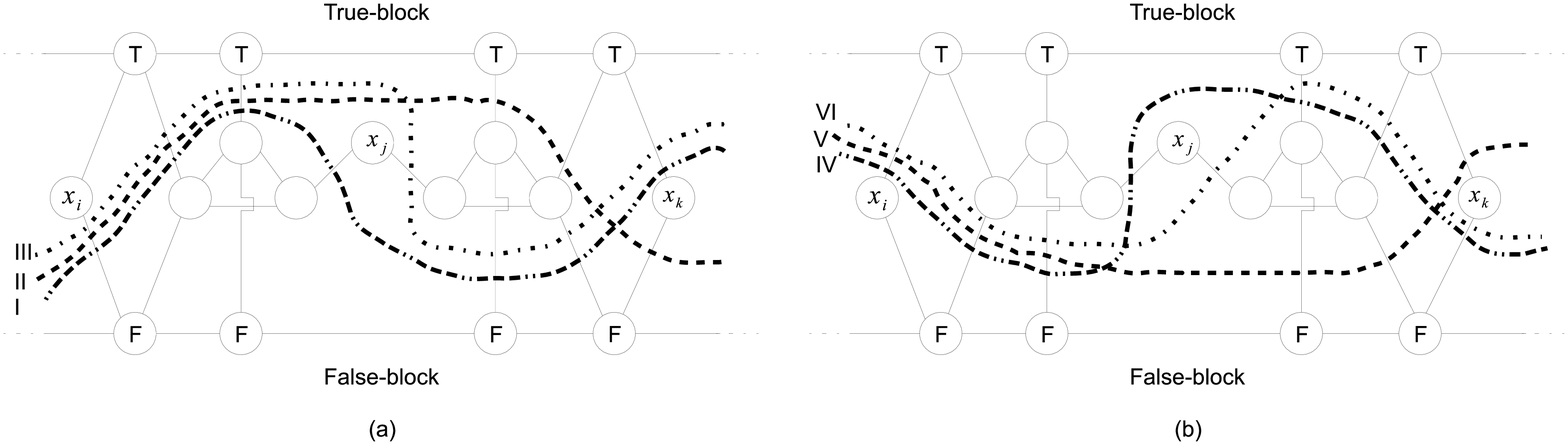}
\caption{The six allowed cuts through the clause-gadget shown in Figure \ref{fig:var-clause-gadget}-(b) that result in two 1-reachable but not 2-reachable sets.}
\label{fig:6clause-gadget}
\end{figure}

Moreover, in case (ii) suppose all three literal-nodes of a clause-gadget are in set $\mathcal{A}$ (the case that all three literal-nodes of a clause-gadget are in $\mathcal{B}$ can be handled via identical arguments). Since node 5 (in Figure~\ref{fig:5clause-gadget-v2}) is in set $\mathcal{A}$, then to respect the properties of cut $\mathcal{C}$, at least one of its neighbors should also lie in $\mathcal{A}$. Without loss of generality due to symmetry, assume that node 4 is in $\mathcal{A}$. This implies that node 2 is also in $\mathcal{A}$; otherwise, cut $\mathcal{C}$ results in a 2-reachable set. Finally, since both nodes 1 and 2 are in $\mathcal{A}$, the node in the False-block connected to both of them has two neighbors in $\mathcal{A}$. This contradicts the definition of $\mathcal{C}$; thus $\mathcal{C}$ cannot leave all literal-nodes in a clause-gadget in the same set. The only possible cuts through clause-gadgets for this case are the ones illustrated in Figure~\ref{fig:6clause-gadget}. It can be seen that none of these six cuts leaves all three literal-nodes of a clause-gadget on one side of the cut.  Hence the second property in the lemma also holds.
\end{IEEEproof}

\subsection{Proof of Lemma~\ref{prop:cuts-inter-edges}}
\begin{IEEEproof}
First, note that a literal-node has the same truth value as its corresponding variable-node if and only if they lie on the same side of cut $\mathcal{C}=(\mathcal{A},\mathcal{B})$. In the proof of Lemma~\ref{prop:cuts-var-clause-gadget}, it was shown that the only possible case for the True and False-blocks is the case that all nodes in the True-block are in $\mathcal{A}$ and all nodes in the False-block are in $\mathcal{B}$ (or vice versa).

In this case, assume that there exists a variable-node $x$ in set $\mathcal{A}$ such that its corresponding literal-node lies in set $\mathcal{B}$ (the case that the variable-node is in $\mathcal{B}$ and the corresponding literal-node is in $\mathcal{A}$ can be analyzed similarly). Then, node $x$ has the following two neighbors in set $\mathcal{B}$: its corresponding literal-node and the node in the False-block that it is connected to (see Figure~\ref{fig:var-clause}).
This contradicts the fact that $\mathcal{A}$ is not 2-reachable and therefore is not possible.
\end{IEEEproof}

\subsection{Proof of Lemma~\ref{lem:r-r-r-hard}}
\begin{IEEEproof}
We prove this claim by a reduction from NAE3SAT. Specifically, we show that graph $\mathcal{G}(\phi)$ has a relaxed-1-degree cut if and only if $\phi$ has a solution within the NAE3SAT constraints.

Suppose that $\mathcal{G}(\phi)$ has a relaxed-1-degree cut $\mathcal{C}=(\mathcal{A},\mathcal{V}\backslash \mathcal{A})$. 
By the first part of Lemma~\ref{prop:cuts-var-clause-gadget}, cut $\mathcal{C}$ has to go through the edges of each variable-gadget as depicted in Figure~\ref{fig:cut-var-gadget} and leaves each variable-node and its negation node on opposite sides of the cut, thereby specifying their truth assignments. Also, by the second part of Lemma~\ref{prop:cuts-var-clause-gadget}, the clause-gadgets are cut by $\mathcal{C}$ according to one of the six cases illustrated in Figure~\ref{fig:6clause-gadget}, which results in having at least one `True' and one `False' literal-node in each clause-gadget. Furthermore, note that by Lemma~\ref{prop:cuts-inter-edges}, the intermediate edges are never cut by $\mathcal{C}$. Hence, all the literal-nodes corresponding to the same variable-node are left in the same set as that variable-node and the negated literal-nodes are in the other set. Consequently, if $\mathcal{G}(\phi)$ has a relaxed-1-degree cut, then $\phi$ is satisfiable within the NAE3SAT constraints.

\begin{table}
\centering
    \begin{tabular}{|c|c|c|c|}
      \hline
      $x_i$ & $x_j$ & $x_k$ & Cut \\
      \hline
      T & T & T & no cut \\
      T & T & F & I \\
      T & F & T & II \\
      T & F & F & III \\
      F & F & T & IV \\
      F & T & F & V \\
      F & T & T & VI \\
      F & F & F & no cut \\
      \hline
    \end{tabular}
    \caption{The truth assignments corresponding to different cuts in a clause-gadget demonstrated in Figure~\ref{fig:6clause-gadget}.}
    \label{tab:6clause-gadget}
\end{table}

On the other hand, if $\phi$ has a solution under the NAE3SAT constraints, then a cut $\mathcal{C}=(\mathcal{A},\mathcal{V}\backslash \mathcal{A})$ can be found in $\mathcal{G}(\phi)$ such that (i) each variable-gadget is cut so that the variable-node and its negation node are connected to the blocks labeled with their truth values, and (ii) each clause-gadget is cut according to its truth assignment as illustrated in Table~\ref{tab:6clause-gadget}. It can be easily observed that using this cut, no node of graph $\mathcal{G}(\phi)$ is incident with more than one edge of the cut-set and hence $\mathcal{G}(\phi)$ has a relaxed-1-degree cut.
This proof of the $\NP$-hardness of the relaxed-1-degree cut problem, together with the fact that this problem is in $\NP$, shows that the relaxed-1-degree cut problem is $\NP$-complete.
\end{IEEEproof}


\subsection{Proof of Theorem~\ref{thm:robustness-hard}}

\begin{IEEEproof}
We show that the 1-degree cut problem is $\NP$-hard by showing that $\mathcal{H}(\phi)$ has a 1-degree cut if and only if $\mathcal{G}(\phi)$ has a relaxed-1-degree cut for any instance $\phi$ of NAE3SAT. It can be easily seen that if $\mathcal{G}(\phi)$ has a relaxed-1-degree cut then $\mathcal{H}(\phi)$ also has a relaxed-1-degree cut (e.g., simply replicate the cut in $\mathcal{G}(\phi)$ for each box in $\mathcal{H}(\phi)$) and thus a 1-degree cut. It only remains to show if $\mathcal{H}(\phi)$ has a 1-degree cut then $\mathcal{G}(\phi)$ has a relaxed-1-degree cut. Assume that sets $\mathcal{A},\mathcal{B}$ and $\mathcal{X}$ partition the nodes of $\mathcal{H}(\phi)$ such that (i) $\mathcal{A}$ and $\mathcal{B}$ are nonempty, and (ii) each node in $\mathcal{A}$ and $\mathcal{B}$ has at most one neighbor outside its own set (i.e., $\mathcal{A}$, $\mathcal{B}$ and $\mathcal{X}$ specify a 1-degree cut).

First, for any clique in $\mathcal{H}(\phi)$ with at least three nodes, the fact that $\mathcal{A}$ and $\mathcal{B}$ are not 2-reachable implies the following two properties:
\begin{enumerate}
  \item Set $\mathcal{X}$ can only contain zero, one, or all nodes of the clique, and
  \item If a node of the clique is in $\mathcal{A}$ (resp. $\mathcal{B}$), then no node of that clique is in $\mathcal{B}$ (resp. $\mathcal{A}$).
\end{enumerate}

Let $\mathcal{T}$ and $\mathcal{F}$ denote the set of all nodes in the True and False-blocks of graph $\mathcal{H}(\phi)$. Several different scenarios can take place for sets $\mathcal{T}$ and $\mathcal{F}$ with respect to sets $\mathcal{A},\mathcal{B}$ and $\mathcal{X}$. First, consider the case that both $\mathcal{T}$ and $\mathcal{F}$ are subsets of $\mathcal{A}$ (the case that both $\mathcal{T}$ and $\mathcal{F}$ are subsets of $\mathcal{B}$ can be analyzed similarly). Since each box of $\mathcal{H}(\phi)$ is isomorphic to $\mathcal{G}(\phi)$, by the same argument as in the proof of Lemma~\ref{prop:cuts-var-clause-gadget} this scenario is not possible as it would leave set $\mathcal{B}$ empty. Now, by property (2) stated above and without loss of generality due to symmetry, assume that $\mathcal{T}\subseteq \mathcal{A}\cup \mathcal{X}$ and $\mathcal{F}\subseteq \mathcal{B}\cup \mathcal{X}$.
If $\mathcal{T} \subseteq \mathcal{X}$ or $\mathcal{F} \subseteq \mathcal{X}$, the same argument as above yields that $\mathcal{A}$ or $\mathcal{B}$ would be empty, respectively.  Therefore, by property (1) above, $|\mathcal{T}\cap \mathcal{X}|\le 1$ and $|\mathcal{F}\cap \mathcal{X}|\le 1$.  Consequently, there exist at least two boxes in $\mathcal{H}(\phi)$ whose True-blocks are subsets of $\mathcal{A}$, and at least two boxes whose False-blocks are subsets of $\mathcal{B}$.  By the pigeonhole principle, there exists a box in $\mathcal{H}(\phi)$, denoted by $\mathcal{G}'(\phi)$, such that its True-block is a subset of $\mathcal{A}$ and its False-block is a subset of $\mathcal{B}$.  We show that no node of $\mathcal{G}'(\phi)$ can be in set $\mathcal{X}$.

\begin{figure}[ht]
\centering
\includegraphics[scale=0.5]{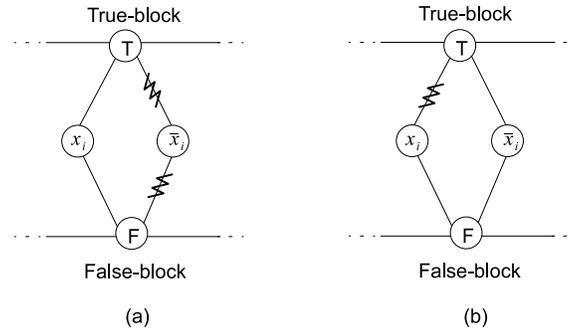}
\caption{Here, the True-block is a subset of $\mathcal{A}$, the False-block is a subset of $\mathcal{B}$ and the variable-node corresponding to $\bar{x}_i$ is in $\mathcal{X}$. Figure (a) shows the edges connected to node $\bar{x}_i$ that are cut. In Figure (b), without loss of generality, it is assumed that node $x_i$ is either in $\mathcal{B}$ or $\mathcal{X}$ and its incident edges with the other endpoints in $\mathcal{A}$ are marked. It can be seen that the node in the True-block has two adjacent nodes outside $\mathcal{A}$.}
\label{fig:noX}
\end{figure}

Suppose that there exists a node in a variable-gadget in $\mathcal{G}'(\phi)$ that lies in $\mathcal{X}$. Then the two nodes in the True and False-blocks connected to that node have a neighbor outside their containing sets, i.e., a neighbor in $\mathcal{X}$ (Figure~\ref{fig:noX}-(a)). Moreover, the other variable-node in that variable-gadget is in $\mathcal{A},\mathcal{B}$ or $\mathcal{X}$. In all these cases this node is in a different set from at least one of its two neighbors in the True and False-blocks (Figure~\ref{fig:noX}-(b)). Hence, there exists at least one node in $\mathcal{A}$ or $\mathcal{B}$ that has two neighbors outside $\mathcal{A}$ or $\mathcal{B}$, respectively, contradicting the fact that we are considering a $1$-degree cut. Therefore, no node in the variable-gadgets can be in set $\mathcal{X}$.

Furthermore, observe that since the True-block is a subset of $\mathcal{A}$ and the False-block is a subset of $\mathcal{B}$, then each variable-node has at least one neighbor outside its containing set. Since it was assumed that each node in $\mathcal{A}$ and $\mathcal{B}$ has at most one neighbor outside its set, it follows that all other neighbors of a variable-node should lie in the same set as that node. Therefore, in $\mathcal{G}'(\phi)$, the endpoints of all intermediate edges, i.e., the edges connecting literal-nodes in the clause-gadgets to the corresponding variable-nodes, lie in the same sets. This, in combination with the fact that none of the variable-nodes in $\mathcal{G}'(\phi)$ are in $\mathcal{X}$ shows that no literal-node in any clause-gadget of $\mathcal{G}'(\phi)$ is in $\mathcal{X}$. It only remains to show that the non-literal-nodes in the clause-gadgets of $\mathcal{G}'(\phi)$, i.e., nodes labeled 2, 3, 4, 6, 7 and 8 as in Figure~\ref{fig:5clause-gadget-v2}, are not in $\mathcal{X}$ either. By the same argument as for variable-nodes, nodes labeled 2 and 8 cannot lie in $\mathcal{X}$. Also, note that since the True and False-blocks of $\mathcal{G}'(\phi)$ are subsets of $\mathcal{A}$ and $\mathcal{B}$, respectively, and node 2 (resp. node 8) is either in $\mathcal{A}$ or $\mathcal{B}$, then one of the edges connecting node 2 (resp. node 8) to the True and False-blocks is excised. Therefore, all its other neighbors, i.e., nodes 3 and 4 (resp. nodes 6 and 7), should lie in the same set as node 2 (resp. node 8). As a result, nodes 3, 4, 6 and 7 cannot be in $\mathcal{X}$. Consequently, no node of $\mathcal{G}'(\phi)$ lies in $\mathcal{X}$.

We have thus shown that if $\mathcal{H}(\phi)$ has a 1-degree cut then $\mathcal{G}'(\phi)$ has a relaxed-1-degree cut. Since $\mathcal{G}(\phi)$ and $\mathcal{G}'(\phi)$ are isomorphic, graph $\mathcal{G}(\phi)$ also has a relaxed-1-degree cut.
Consequently, $\mathcal{H}(\phi)$ has a 1-degree cut if and only if $\mathcal{G}(\phi)$ has a relaxed-1-degree cut for any NAE3SAT instance $\phi$.
Hence, the 1-degree cut problem is $\NP$-hard, and thus $\NP$-complete by virtue of being in $\NP$.
\end{IEEEproof}

\subsection{Proof of Lemma~\ref{lem:r-r-deg-cut}}

\begin{IEEEproof}
Recall that in order to show that the relaxed-1-degree cut problem is $\NP$-hard, for any NAE3SAT instance $\phi$ we constructed a graph $\mathcal{G}(\phi)=\{\mathcal{V},\mathcal{E}\}$ such that $\mathcal{G}(\phi)$ had a relaxed-1-degree cut if and only if $\phi$ was satisfiable within the NAE3SAT constraints. Here, for any $\rho\in \mathbb{Z}_{\ge 1}$, we make two modifications to $\mathcal{G}(\phi)$ to construct graph $\mathcal{G}_{\rho}(\phi)=\{\mathcal{V}_{\rho},\mathcal{E}_{\rho}\}$ so that $\mathcal{G}_{\rho}(\phi)$ has a relaxed-$\rho$-degree cut if and only if the original graph $\mathcal{G}(\phi)$ has a relaxed-1-degree cut. We can then conclude that the relaxed-$\rho$-degree cut problem is $\NP$-hard and thus $\NP$-complete by virtue of being in $\NP$.

The two modifications to graph $\mathcal{G}(\phi)$ are: (i) a modification to the nodes in the True and False-blocks of $\mathcal{G}(\phi)$, and (ii) a modification to the other nodes of $\mathcal{G}(\phi)$, i.e. nodes in the variable and clause-gadgets. In modification (i), for each node $v$ in the True-block (resp. False-block) of $\mathcal{G}(\phi)$, we add $\rho-1$ nodes in the False-block (resp. True-block) and connect them to $v$. Thus, this step adds a total of $2(\rho-1)(4m+t)$ nodes to the graph. In modification (ii), for each node $u$ in the variable and clause-gadgets of $\mathcal{G}(\phi)$ that is not in the True or False-blocks, we add $\rho-1$ nodes in each of the True and False-blocks and connect $u$ to them. This step adds a total of $2(\rho-1)(9m+2t)$ nodes to the graph. The new graph is called $\mathcal{G}_{\rho}(\phi)$. Note that the nodes added to the True and False-blocks of $\mathcal{G}(\phi)$ to construct $\mathcal{G}_{\rho}(\phi)$ are connected to all other nodes in those blocks and hence the True and False-blocks of $\mathcal{G}_{\rho}(\phi)$ are complete subgraphs (each containing $(4m+t)\rho + (9m+2t)(\rho-1)$ nodes).  Figure~\ref{fig:G_2-phi} demonstrates $\mathcal{G}_{2}(\phi)$ for the graph $\mathcal{G}(\phi)$ shown in Figure~\ref{fig:var-clause}.

\begin{figure}[ht]
\centering
\includegraphics[scale=0.3]{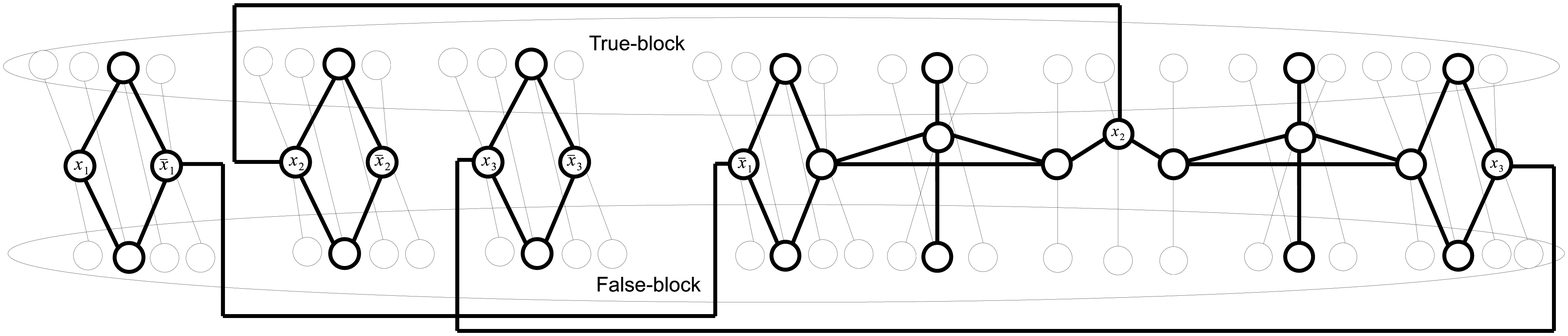}
\caption{Graph $\mathcal{G}_{2}(\phi)$ constructed from the graph $\mathcal{G}(\phi)$ demonstrated in Figure~\ref{fig:var-clause}. The highlighted nodes and edges are the ones that exist in $\mathcal{G}(\phi)$ as well.}
\label{fig:G_2-phi}
\end{figure}

It is easy to observe that if $\mathcal{G}(\phi)$ has a relaxed-1-degree cut $\mathcal{C}=(\mathcal{A}',\mathcal{B}')$, then $\mathcal{G}_{\rho}(\phi)$ has a relaxed-$\rho$-degree cut $\mathcal{C}_{\rho}=(\mathcal{A},\mathcal{B})$, by assigning the nodes corresponding to those in $\mathcal{A}'$ and $\mathcal{B}'$ in $\mathcal{G}(\phi)$ to sets $\mathcal{A}$ and $\mathcal{B}$ in $\mathcal{G}_{\rho}(\phi)$, respectively. Furthermore, each newly added node in the True and False-blocks of $\mathcal{G}_{\rho}(\phi)$ should be assigned to the same set as the rest of the nodes in its block (recall that by Lemma~\ref{prop:cuts-TFblocks}, in the relaxed-1-degree cut problem all the nodes in each of the True and False-blocks belong to the same subset of nodes). So, it remains to prove the converse statement.

If graph $\mathcal{G}_{\rho}(\phi)$ has a relaxed-$\rho$-degree cut, then there exists a cut $\mathcal{C}_{\rho}=(\mathcal{A},\mathcal{V}_{\rho}\backslash\mathcal{A})$ that partitions the nodes of $\mathcal{G}_{\rho}(\phi)$ into two nonempty sets $\mathcal{A}$ and $\mathcal{B}=\mathcal{V}_{\rho}\backslash \mathcal{A}$ such that no node of the graph has more than $\rho$ neighbors outside its containing set.   Note that since the True (resp. False) block is a clique with more than $2\rho + 1$ nodes for $\rho \in \mathbb{Z}_{\ge 1}$, the relaxed-$\rho$-degree cut cannot separate the nodes in the True (resp. False) block.

Next, note that the True and False-blocks of $\mathcal{G}_{\rho}(\phi)$ cannot both lie in set $\mathcal{A}$ (or similarly $\mathcal{B}$).  For if they did, it can be easily checked that except for the nodes labeled 4, 5 and 6 in Figure~\ref{fig:5clause-gadget-v2}, if any other node of $\mathcal{G}_{\rho}(\phi)$ lies in $\mathcal{B}$, then it has at least $2\rho$ neighbors outside its containing set in the True and False-blocks and hence violates the definition of cut $\mathcal{C}_{\rho}$. Now nodes labeled 4 and 6 cannot lie in $\mathcal{B}$, since they have $2\rho$ neighbors in set $\mathcal{A}$. Consequently, node 5 should also lie in $\mathcal{A}$ which makes $\mathcal{B}$ empty and hence is not allowed.

Without loss of generality, we now assume that the True-block of $\mathcal{G}_{\rho}(\phi)$ is a subset of
$\mathcal{A}$ and its False-block is a subset of $\mathcal{B}$. In the following, we prove that in this case, cut $\mathcal{C}_{\rho}$ has three properties.

First, similar to the relaxed-1-degree cut problem, both a variable-node and its negation node in a variable-gadget cannot lie in set $\mathcal{A}$ (resp. $\mathcal{B}$); otherwise, the node in the False-block (resp. True-block) connected to both these nodes has $\rho+1$ neighbors outside its set.

Second, since a variable-node and its negation node in a variable-gadget lie in different sets, each of these nodes are incident with $\rho$ edges in the cut-set. Therefore, no other edge connected to these nodes can be excised by cut $\mathcal{C}_{\rho}$. In particular, the intermediate edges connecting literal-nodes in clause-gadgets to their corresponding variable-nodes should be left uncut.

Third, we argue that the three literal-nodes in a clause-gadget cannot be in the same subset of nodes, say $\mathcal{A}$. Suppose that the literal-nodes corresponding to nodes labeled 1 and 9 in Figure~\ref{fig:5clause-gadget-v2} are in set $\mathcal{A}$. Due to the same argument as for variable-gadgets, the nodes that share a neighbor in the True and False-blocks with these nodes, i.e., the nodes labeled 2 and 8 in Figure~\ref{fig:5clause-gadget-v2}, should lie in $\mathcal{B}$. Then nodes labeled 3, 4, 6 and 7 in Figure~\ref{fig:5clause-gadget-v2} should also be in $\mathcal{B}$. Now if node 5 in Figure~\ref{fig:5clause-gadget-v2} lies in $\mathcal{A}$, then it has at least $\rho+1$ neighbors outside its containing set, which is not admissible.

The above three properties force $\mathcal{C}_{\rho}$ to partition the nodes of $\mathcal{G}_{\rho}(\phi)$ into two sets $\mathcal{A}$ and $\mathcal{B}$ such that sets $\mathcal{A}'=\mathcal{V}\cap \mathcal{A}$ and $\mathcal{B}'=\mathcal{V}\cap \mathcal{B}$ are nonempty and partition the nodes of $\mathcal{G}(\phi)$ into one of the forms shown in Figures~\ref{fig:cut-var-gadget} and~\ref{fig:6clause-gadget}, without cutting through the True or False-blocks of $\mathcal{G}(\phi)$. Therefore, if graph $\mathcal{G}_{\rho}(\phi)$ has a relaxed-$\rho$-degree cut then $\mathcal{G}(\phi)$ has a relaxed-1-degree cut and hence the relaxed-$\rho$-degree cut problem is $\NP$-hard.
\end{IEEEproof}


\subsection{Proof of Theorem~\ref{thm:r-deg-cut}}

\begin{IEEEproof}
As we did in the $\NP$-hardness proof of the 1-degree cut problem, for the $\rho$-degree cut problem we build graph $\mathcal{H}_{\rho}(\phi)$ by taking $2\rho+1$ copies of $\mathcal{G}_{\rho}(\phi)$ and adding edges to construct complete subgraphs on all the nodes in the True and False-blocks of these $2\rho+1$ copies. Again, each copy of $\mathcal{G}_{\rho}(\phi)$ in $\mathcal{H}_{\rho}(\phi)$ is called a box.

We prove that the $\rho$-degree cut problem is $\NP$-hard by showing that $\mathcal{H}_{\rho}(\phi)$ has a $\rho$-degree cut if and only if $\mathcal{G}_{\rho}(\phi)$ has a relaxed-$\rho$-degree cut. It is not hard to see that if $\mathcal{G}_{\rho}(\phi)$ has a relaxed-$\rho$-degree cut, then $\mathcal{H}_{\rho}(\phi)$ has a $\rho$-degree cut and hence it remains to show the converse statement also holds.

Assume that $\mathcal{H}_{\rho}(\phi)$ has a $\rho$-degree cut, and that sets $\mathcal{A}$, $\mathcal{B}$ and $\mathcal{X}$ partition the nodes of $\mathcal{H}_{\rho}(\phi)$ with $\mathcal{A}$ and $\mathcal{B}$ nonempty and at most $\rho$-reachable.  
It can be easily observed that the following generalized version of the two properties in the proof of Theorem~\ref{thm:robustness-hard} holds for cliques with at least $2\rho+1$ nodes:

\begin{enumerate}
  \item Set $\mathcal{X}$ can contain up to $\rho$ nodes or all nodes of the clique, and
  \item If a node of the clique is in $\mathcal{A}$ (resp. $\mathcal{B}$), then no node of that clique is in $\mathcal{B}$ (resp. $\mathcal{A}$).
\end{enumerate}

We denote the True and False-blocks of $\mathcal{H}_{\rho}(\phi)$ by $\mathcal{T}$ and $\mathcal{F}$, respectively. As in the proof of Theorem~\ref{thm:robustness-hard}, it can be argued that the only possible cases for $\mathcal{T}$ and $\mathcal{F}$ are $\mathcal{T}\subseteq \mathcal{A}\cup \mathcal{X}$ and $\mathcal{F}\subseteq \mathcal{B}\cup \mathcal{X}$, or $\mathcal{T}\subseteq \mathcal{B}\cup \mathcal{X}$ and $\mathcal{F}\subseteq \mathcal{A}\cup \mathcal{X}$. The analysis of the latter case is removed due to symmetry. The True-block cannot be a subset of $\mathcal{X}$, as it would leave set $\mathcal{A}$ empty.  Similarly, the False-block cannot be a subset of $\mathcal{X}$, as it would leave set $\mathcal{B}$ empty.  Thus by property (1), $|\mathcal{T}\cap \mathcal{X}|\le \rho$ and so there exist $\rho+1$ boxes in $\mathcal{H}_{\rho}(\phi)$ whose True-blocks are subsets of $\mathcal{A}$. The same argument holds for the False-blocks of the boxes in $\mathcal{H}_{\rho}(\phi)$. Consequently, there is a box in $\mathcal{H}_{\rho}(\phi)$ whose True and False-blocks are subsets of $\mathcal{A}$ and $\mathcal{B}$, respectively. This box is denoted by $\mathcal{G}'_{\rho}(\phi)$.  We show that no node in $\mathcal{G}'_{\rho}(\phi)$ can lie in set $\mathcal{X}$.

Suppose that a node in a variable-gadget of $\mathcal{G}'_{\rho}(\phi)$ lies in $\mathcal{X}$. Regardless of the set the other node in that variable-gadget is in, either the node in the True-block or the one in the False-block to which both these variable-nodes are connected has $\rho+1$ neighbors outside its set, which is not allowed (see Figure~\ref{fig:noX-G_2} for an illustration when $\rho = 2$).

\begin{figure}[ht]
\centering
\includegraphics[scale=0.45]{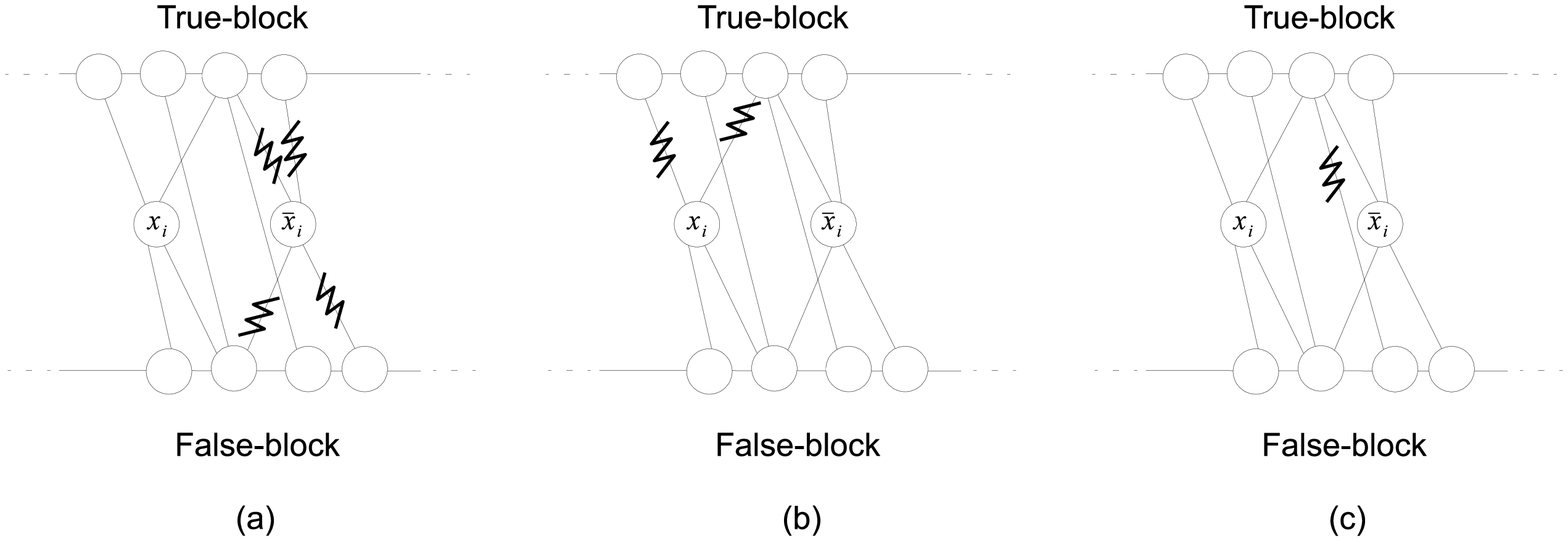}
\caption{The True and False-blocks are subsets of $\mathcal{A}$ and $\mathcal{B}$, respectively, and the variable-node corresponding to $\bar{x}_i$ is in $\mathcal{X}$. Figure (a) shows the edges incident to node $\bar{x}_i$ that are excised. In Figure (b), without loss of generality, it is assumed that node $x_i$ is either in set $\mathcal{B}$ or $\mathcal{X}$ and the edges that connect the nodes in $\mathcal{A}$ to it are depicted. Figure (c) demonstrates another edge that connects a node in $\mathcal{A}$ that is in the True-block to a node in $\mathcal{B}$ that is in the False-block. It can be seen that the node in the True-block connected to both the nodes in the variable-gadget has three neighbors outside its set.}
\label{fig:noX-G_2}
\end{figure}

Moreover, since the True and False-blocks of $\mathcal{G}'_{\rho}(\phi)$ are subsets of $\mathcal{A}$ and $\mathcal{B}$, respectively, each variable-node has $\rho$ neighbors outside its set and hence cannot afford to be connected to another node outside its set. Therefore, the intermediate edges connecting literal-nodes to their corresponding variable-nodes cannot be excised. This in combination with the fact that no variable-node is in $\mathcal{X}$ shows that no literal-node in any clause-gadget is in $\mathcal{X}$. Also, by the same argument provided for variable-gadgets, it can be shown that the nodes labeled 2 and 8 in Figure~\ref{fig:5clause-gadget-v2} cannot lie in set $\mathcal{X}$. Moreover, the nodes labeled 2, 3 and 4 in Figure~\ref{fig:5clause-gadget-v2} (resp. nodes labeled 6, 7 and 8) should lie in the same set; otherwise, node 2 (resp. node 8) has more than $\rho$ neighbors outside its set. Thus, no node in any clause-gadget can be in $\mathcal{X}$. As a result, no node in $\mathcal{G}'_{\rho}(\phi)$ is in $X$.

Thus, if $\mathcal{H}_{\rho}(\phi)$ has a $\rho$-degree cut, then $\mathcal{G}'_{\rho}(\phi)$ and hence $\mathcal{G}_{\rho}(\phi)$ has a relaxed-$\rho$-degree cut, due to the fact that $\mathcal{G}'_{\rho}(\phi)$ and $\mathcal{G}_{\rho}(\phi)$ are isomorphic. 
Therefore, the converse statement is proved which concludes the $\NP$-hardness of the $\rho$-degree cut problem.  Since the $\rho$-degree cut problem is in $\NP$, it can be concluded that it is $\NP$-complete.
\end{IEEEproof}

\subsection{Proof of Lemma~\ref{lem:hardness-apx}}

\begin{IEEEproof}
Assume that some polynomial-time algorithm $\mathrm{ALG}$ provides an approximation ratio $1< \alpha < 2$ for all graphs $\mathcal{G}$.  Consider the set of connected graphs that have a 1-degree cut; on any graph $\mathcal{G}$ from this set, the output of the algorithm must satisfy $\mathrm{ALG}(\mathcal{G}) \le \alpha\mathrm{OPT}(\mathcal{G}) < 2$, i.e., $\mathrm{ALG}(\mathcal{G}) = 1$.  Thus, given any connected graph $\mathcal{G}$, the algorithm would output $\mathrm{ALG}(\mathcal{G}) = 1$ if and only if the graph has a $1$-degree cut, contradicting the fact that this is an $\NP$-hard problem (as shown in Theorem~\ref{thm:robustness-hard}).  Therefore, an approximation ratio less than 2 is not obtainable for the $\rho$-degree cut problem unless $\Poly=\NP$.
\end{IEEEproof}

\begin{remark}
Note that this proof holds for any integer-valued optimization problem whose decision version is $\NP$-hard with parameter equal to $1$.
\end{remark}

\end{appendix}

\bibliographystyle{siam}
\bibliography{refs}
\end{document}